\documentclass{article}
\usepackage[latin1]{inputenc}
\usepackage[english]{babel}
\usepackage[dvips]{graphics}
\usepackage{epsfig}
\usepackage{fullpage}
\usepackage{amsfonts}
\usepackage{multicol}
\usepackage{latexsym,amssymb}     
\usepackage{amsmath}
\usepackage{indentfirst}
\usepackage{makeidx}
\usepackage{theorem}
\usepackage{bbold}
\usepackage{psfrag}
\usepackage{enumitem}

\newcommand{\cH}{\mathcal{H}}
\newcommand{\rr}{\mathbb{R}}
\newcommand{\ee}{\mathbb{E}}

\newcommand{\nn}{\mathbb{N}}
\newcommand{\pp}{\mathbb{P}}
\newcommand{\ov}{\operatornamewithlimits}
\newcommand{\ds}{\displaystyle}

\def\BS{Black \& Scholes }

\newcommand{\blanc}[1]{\vspace{#1\baselineskip}}

\newtheorem{nt_theorem}{Theorem}
\newenvironment{theorem}{\blanc{1.5}\begin{nt_theorem}---}{\end{nt_theorem}\blanc{1.5}}

\newtheorem{nt_proposition}[nt_theorem]{Proposition}

\newtheorem{nt_corollaire}[nt_theorem]{Corollary}
\newenvironment{Corollary}{\blanc{1.5}\begin{nt_corollaire}---}{\end{nt_corollaire}\blanc{1.5}}

\newtheorem{nt_definition}[nt_theorem]{Definition}

\newtheorem{nt_lemma}[nt_theorem]{Lemma}
\newenvironment{lemma}{\blanc{0.5}\begin{nt_lemma}---}{\end{nt_lemma}\blanc{0.5}}

\newtheorem{nt_jeulemma}[nt_theorem]{Jeulin's Lemma}

\newtheorem{nt_conjecture}[nt_theorem]{Conjecture}

\newtheorem{nt_remark}[nt_theorem]{Remark}

\newenvironment{proof}{{\textit{Proof : }}}{\hfill$\Box$
\\\bigskip}
\newenvironment{proofOFT1}{{\textit{Proof of Theorem 1 : }}}{\hfill$\Box$
\\\bigskip}
\newenvironment{proofOFT2}{{\textit{Proof of Theorem 2 : }}}{\hfill$\Box$
\\\bigskip}

\makeatletter
\newcounter{hypo}

\newcommand*{\dohypo}{\textbf{(${\mathcal H}$\thehypo)}}
\newenvironment{hypo}[1][]{%
  
  \refstepcounter{hypo}
  \list{}{%
    \settowidth{\labelwidth}{\dohypo}%
    \setlength{\labelsep}{10pt}%
    \setlength{\leftmargin}{\labelwidth}
    \advance\leftmargin\labelsep%
  }%
\item[\dohypo  #1]%
}{%
  \endlist
}

\def\hypref#1{\hyperref[hyp:#1]{(${\mathcal H}$\ref*{hyp:#1})}}
\def\hypreff#1#2{\hyperref[hyp:#2]{(${\mathcal H}$\ref*{hyp:#1}-{\it
\ref*{hyp:#2})}}}

\makeatother

\begin{document}

\vspace{5mm}

\begin{center}
\textbf{\Large{Coupling Index and stocks}}\\
\end{center}

$\,$

\begin{center}
\emph{Benjamin Jourdain}\footnote{Université Paris-Est, CERMICS,
Projet MathFi ENPC-INRIA-UMLV. This research benefited from the
support of the "Chair Risques Financiers", Fondation du Risque.

Postal address : 6-8 av. Blaise Pascal, Cité Descartes,
Champs-sur-Marne, 77455 Marne-la-Vallée Cedex 2.

E-mails : \underline{jourdain@cermics.enpc.fr} and
\underline{sbai@cermics.enpc.fr} } and \emph{Mohamed Sbai$\,^1$}
\end{center}

\vspace{20mm}

\begin{abstract}
In this paper, we are interested in continuous time models in
which the index level induces some feedback on the dynamics of its
composing stocks. More precisely, we propose a model in which the
log-returns of each stock may be decomposed into a systemic part
proportional to the log-returns of the index plus an idiosyncratic
part. We show that, when the number of stocks in the index is large,
this model may be approximated by a local volatility model for the
index and a stochastic volatility model for each stock with
volatility driven by the index. This result is useful in a calibration
perspective : it suggests that one should first calibrate the local
volatility of the index and then calibrate the dynamics of each stock. We explain how
to do so in the limiting simplified model and in the original model.
\end{abstract}

 \vspace{20mm}

\section*{Introduction}
{F}rom the early eighties, when trading on stock index was
introduced, quantitative finance faced the problem of efficiently
pricing and hedging index options along with their underlying
components. Many advances have been made for single stock modeling
and a variety of solutions to escape from the very restrictive \BS
model has been deeply investigated (such as local volatility models,
models with jumps or stochastic volatility models). However, when
the number of underlyings is large, index option pricing, or more
generally basket option pricing, remains a challenge unless one
simply assumes constantly correlated dynamics for the stocks. The
problem then is the impossibility of fitting both the stocks and the
index smiles.

We try to address this issue by making the dynamics of the stocks
depend on the index. The natural fact that the volatility of the
index is related to the volatilities of its underlying components
has already been accounted for in the works of Avellaneda \textit{et
al.}~\cite{Avellaneda} and Lee \textit{et al.}~\cite{Lee}. In the
first paper, the authors use a large deviation asymptotics valid for
small values of the product of the maturity by the square of
the volatility to
reconstruct the local volatility of the index from the local
volatilities of the stocks. They express this dependence in terms of
the implied volatilities using the results of Berestecky \textit{et
al.}(\cite{Berestycki1},\cite{Berestycki2}). In the second paper,
the authors reconstruct the Gram-Charlier expansion of the
probability density of the index from the stocks using a
moments-matching technique. Both papers consider local volatility
models for the stocks and a constant correlation matrix but the
generalization to stochastic volatility models or to varying
correlation coefficients is not straightforward.

Another point of view is to say that the volatility of a composing
stock should be related to the index level, or say to the volatility
of the index, in some way. This is not astonishing since the index
represents the move of the market and reflects the view of the
investors on the state of the economy. Moreover, it is consistent with
equilibrium economic models like CAPM. Following this idea, we
propose a new modeling framework in which the volatility of the
index and the volatilities of the stocks are related. We show that,
when the number of underlying stocks tends to infinity, our model
reduces to a local volatility model for the index and to a
stochastic volatility model for the stocks where the stochastic
volatility depends on the index level. This asymptotics is
reasonable since the number of stocks composing an index is usually large.
As a consequence, the correlation matrix between the stocks in our
model is not constant but stochastic and we show that it is consistent
with empirical studies. Finally, we address calibration issues and
we show that it is possible, within our framework, to fit both index
and stocks smiles. The method we introduce is based on the
simulation of SDEs nonlinear in the sense of McKean, and
non-parametric estimation of conditional expectations.

This paper is organized as follows. In Section 1, we specify our
model for the index and its composing stocks and in Section 2 we
study the limiting model when the number of underlying stocks goes
to infinity. Section 3 is devoted to calibration issues. Numerical
results are presented in Section 4 and the conclusion is given in
Section 5.

\vspace{2mm}

$\,$\\
\underline{Acknowledgements:} We thank Lorenzo Bergomi, Julien Guyon
and all the equity quantitative research team of Societe Generale
CIB for numerous fruitful discussions and for providing us with the
market data.

\vspace{5mm}

$\,$\\

\section{Model Specification}

An index is a collection of stocks that reflects the performance of
a whole stock market or a specific sector of a market. It is valued
as a weighted sum of the value of its underlying components. More
precisely, if $I^M_t$ stands for the value at time $t$ of an index
composed of $M$ underlyings, then
\begin{equation}
I^M_t=\sum_{j=1}^M w_j S^{j,M}_t,
\end{equation}
where $S^{j,M}_t$ is the value of the stock $j$ at time $t$ and the
weightings $(w_j)_{j=1\dots M}$ are given constants\footnote{The
weightings are periodically updated but, as usually assumed, we
suppose that, up to maturities of the options considered, they are
constant. When updated, they are often chosen proportional to the
market capitalizations of the stocks.}.

Unless otherwise stated, we always work under a risk-neutral
probability measure. In order to account for the influence of the
index on its underlying components, we specify the following
stochastic differential equations for the stocks
\begin{equation}
\forall j \in \{1,\dots,M\}, \quad
\frac{dS^{j,M}_t}{S^{j,M}_t}=(r-\delta_j) dt + \beta_j
\,\sigma(t,I^{M}_t) dB_t + \eta_j(t,S^{j,M}_t) dW^j_t,\;S^{j,M}_0=s^j_0
\end{equation}
where
\begin{itemize}
\item $r$ is the short interest rate,
\item $s^j_0$ is the initial value of the stock $j$,
\item $\delta_j \in [0,\infty[$ is the continuous dividend rate of the stock $j$,
\item $\beta_j$ is the usual beta coefficient of the stock $j$ that
quantifies the sensitivity of the stock returns to the index returns
(see the seminal paper of Sharpe~\cite{sharpe}). It is defined as
$\frac{Cov(r_{j},r_I)}{Var(r_I)}$ where $r_{j}$ (respectively
$r_{I}$) is the rate of return of the stock $j$ (respectively of the
index),
\item $(B_t)_{t\in[0,T]},(W^1_t)_{t\in[0,T]},\dots,
 (W^M_t)_{t\in[0,T]}$ are independent Brownian motions,
\item the functions $\sigma,\eta_1,\dots,\eta_M:[0,T]\times\rr\to\rr$ satisfy the
usual Lipschitz and growth assumptions that ensure existence and
strong uniqueness of the solutions (see for example Theorem 5.2.9 of
\cite{KaratzasShreve}) :
\begin{hypo}
\label{hyp:LIP} $\exists K$ such that $\forall (t,s_1,s_2) \in
[0,T]\times \rr^M\times \rr^M,$
\[
\begin{array}{l}
\ds \sum_{j=1}^M \left|s^j_1 \sigma\left(t,\sum_{k=1}^M w_k
s^{k}_1\right)\right| + \left|s^j_1
\eta_j(t,s^j_1)\right| \leq K\left(1+|s_1|\right)\\[3mm]
\ds \sum_{j=1}^M \left|s^j_1 \sigma\left(t,\sum_{k=1}^M w_k
s^{k}_1\right)-s^j_2 \sigma\left(t,\sum_{k=1}^M w_k
s^{k}_2\right)\right| \leq K|s_1-s_2|\\[3mm]
\ds \sum_{j=1}^M \left|s^j_1 \eta_j(t,s^j_1)-s^j_2
\eta_j(t,s^j_2)\right|\leq K|s_1-s_2|.
\end{array}
\]
\end{hypo}
 \end{itemize}

As a consequence, the index satisfies the following stochastic
differential equation :
\begin{equation}
dI^M_t=r I^M_t dt -\left(\sum_{j=1}^M \delta_j w_j S^{j,M}_t\right)
dt+ \left(\sum_{j=1}^M \beta_j w_j S^{j,M}_t\right)\sigma(t,I^{M}_t)
dB_t +\sum_{j=1}^M w_j S^{j,M}_t \eta_j(t,S^{j,M}_t) dW^j_t
\label{indexSDE}
\end{equation}

Before going any further, let us make some preliminary remarks on
this framework.

\begin{itemize}
\item[-] We have $M$ coupled stochastic differential equations.
The dynamics of a given stock depends on all the other stocks
composing the index through the volatility term $\sigma(t,I^M_t)$. Since
there are $M$ linearly independent assets and $M+1$ driving Brownian
motions, the market is incomplete.
\item [-] Accounting for the dividends is not relevant for all types
of indices. Indeed, for many performance-based indices (such as the
German DAX index) dividends and other events are rolled into the
final value of the index.
\item[-] The cross-correlations between stocks are not constant but
stochastic :
\[\rho_{ij}(t)=\frac{\beta_i \beta_j \sigma^2(t,I^{M}_t)}{\ds \sqrt{\beta_i^2\sigma^2(t,I^{M}_t)+\eta_i^2(t,S^{i,M}_t)}\,
\sqrt{\beta_j^2\sigma^2(t,I^{M}_t)+\eta_j^2(t,S^{j,M}_t)}}\]
Note that they depend not only on the stocks but also on the index.
More importantly, it is commonly observed that the more the market
is volatile, the more the stocks tend to be highly correlated. This
feature is reproduced here as we can easily check that an increase
in the index volatility, with everything else left unchanged,
produces an increase in the cross-correlations.

In a recent paper, Cizeau \textit{et al.}~\cite{Bouchaud1} show that
it is possible to capture the essential features of stocks
cross-correlations, in particular in extreme market conditions, by a
simple non-Gaussian one factor model. The authors successfully
compare different empirical measures of correlation with the
prediction of the following model :
\begin{equation}
r_j(t)=\beta_j r_I(t)+\epsilon_j(t)\label{bouchaud}
\end{equation}

where $r_j(t)=\frac{S^j_t}{S^j_{t-1}}-1$ is the daily return of
stock $j$, $r_I(t)$ is the daily return of the market and the
residuals $\epsilon_j(t)$ are independent random variables following
a fat-tailed distribution\footnote{The authors have chosen a Student
distribution in their numerical experiments.}.

Our model is in line with~(\ref{bouchaud}). Indeed, since the beta
coefficients are usually narrowly distributed around 1, the factor
$\sum_{j=1}^M \beta_j w_j S^{j,M}_t$  of $\sigma(t,I^M_t)$ in
(\ref{indexSDE}) is close to $I^M_t$. Moreover,  since $$\ee\left(\left(\int_0^T\sum_{j=1}^M w_j
    S^{j,M}_t \eta_j(t,S^{j,M}_t) dW^j_t\right)^2\right)\leq \sum_{j=1}^M
w_j^2\sup_{1\leq j\leq M}\int_0^T\ee\left((S^{j,M}_t
  \eta_j(t,S^{j,M}_t))^2\right)dt\sim \sum_{j=1}^M
w_j^2 T,$$ one can
neglect the term $\sum_{j=1}^M w_j S^{j,M}_t \eta_j(t,S^{j,M}_t) dW^j_t$
in the dynamics of the index when $\sum_{j=1}^M
w_j^2$ is small. Of course, this approximation worsens when the maturity
$T$ increases. The latter condition is
satisfied when $M$ is
large and the weighting vector
$(w_1,\hdots,w_M)$ is close to the vector
$(\frac{1}{M},\hdots,\frac{1}{M})$ with constant coefficients for which $\sum_{j=1}^M\frac{1}{M^2}=\frac{1}{M}$. Then, if we denote by $r_j$ the
log-return of the stock $j$ and by $r_{I^M}$ the log-return of the
index, both on a daily basis, we will have
\[r_j=\beta_j r_{I^M}+\eta_j \Delta W^j + \text{drift},\]
where $\Delta W^j$ is an independent Gaussian noise. Consequently,
in our model too, the return of a stock is decomposed into a
systemic part driven by the index, which represents the market, and
a residual part.
\end{itemize}

\section{Asymptotics for a large number of underlying stocks}
The number of underlying components of an index is usually
large\footnote{500 stocks for the S\&P 500 index, 100 stocks for the
FTSE 100 index, 40 stocks for the CAC40 index, etc.}. As discussed in
the previous section, when $\sum_{j=1}^M w_j^2$ is
small, one can neglect the term $\sum_{j=1}^M w_j S^{j,M}_t
\eta_j(t,S^{j,M}_t) dW^j_t$ in \eqref{indexSDE} and derive a simplified
approximate dynamics for the index.
The aim of this section is to quantify the error we commit by doing
so.

To be specific, consider the limit candidate $(I_t)_{t\in[0,T]}$
solution of the following SDE :
\begin{equation}
\left\{\begin{array}{rcl} dI_t&=&\ds (r-\delta) I_t dt + \beta I_t
\sigma(t,I_t) dB_t \\
I_0&=&\ds i_0

\end{array}\right.
\label{Ilim}
\end{equation}
where $i_0=\sum_{j=1}^M w_js^j_0$ and $\delta$ and $\beta$ are two constant parameters that will be
discussed later.

In the following theorem, we give an upper bound for the
$L^{2p}$-distance between $(I^{M}_t)_{t\in[0,T]}$ and
$(I_t)_{t\in[0,T]}$ under mild assumption on the volatility
coefficients :

\begin{theorem}\label{convind}
Let $p\in{\mathbb N}^*$. Under assumption ($\cH$\ref{hyp:LIP}) and
if the following assumptions on the volatility coefficients hold,
\begin{hypo}
  \label{hyp:bornitude}
$\exists K_b$ such that $\forall (t,s) \in [0,T]\times \rr_+,\quad
|\sigma(t,s)|+|\eta_j(t,s)|\leq K_b$.
\end{hypo}
\begin{hypo}\label{hyp:xsLip} $\exists K_\sigma$ such that $\forall
(t,s_{1},s_{2}) \in [0,T]\times \rr_+\times \rr_+,\quad
|s_{1}\sigma(t,s_{1})-s_{2}\sigma(t,s_{2})|\leq K_\sigma|s_{1}-s_{2}|$.
\end{hypo}

then
\[\ee\left(\sup_{0\leq t \leq T}|I^{M}_t-I_t|^{2p}\right) \leq C_T \left(\left(\sum_{j=1}^{M}
w_{j}^2\right)^{\!\!p} + \left(\sum_{j=1}^M
  w_j|\beta_j-\beta|\right)^{2p}+ \left(\sum_{j=1}^M
w_j|\delta_j-\delta|\right)^{2p} \right)\] where \[\ds C_T=8^{2p-1}
T^p (T^p+K_pK_b^{2p})C_p\exp\left(4^{2p-1}T(2^{2p-1}K_pT^{p-1}(\beta
  K_\sigma)^{2p}+(2T)^{2p-1}
\delta^{2p}+r^{2p} \,T^{2p-1})\right)\] and \[C_p=\max_{1\leq j \leq
M}|s^{j}_{0}|^{2p} \exp\left(\left(2r+(2p-1)(\max_{j \geq
1}\beta_j^2+1)K_b^2\right)pT\right).\]
\end{theorem}
According to this result proved in the appendix, the smaller
$P^M_w\stackrel{\rm def}{=}\sqrt{\sum_{j=1}^Mw_j^2}$,
$P^M_\beta\stackrel{\rm def}{=}\sum_{j=1}^M w_j|\beta_j-\beta|$ and
$P^M_\delta\stackrel{\rm def}{=}\sum_{j=1}^M w_j|\delta_j-\delta|$,
the closer $I$ and $I^M$. The first quantity $P^M_w$ is small when
the weighting vector $(w_1,\hdots,w_M)$ is close to
$(\frac{1}{M},\hdots,\frac{1}{M})$ and $M$ is large. Let us now
discuss how to choose $\beta$ and $\delta$ minimizing $P^M_\beta$
and $P^M_\delta$. Let $Y_\beta$ and $Y_\delta$ be discrete random
variables having the following probability distributions : \[\forall
j \in \{1,\dots,M\},\quad \quad
\pp\left(Y_\beta=\beta_j\right)=\frac{w_j}{\sum_{i=1}^M w_i}\quad
\text{and}\quad
\pp\left(Y_\delta=\delta_j\right)=\frac{w_j}{\sum_{i=1}^M w_i} .\]
Then
\[P^M_\beta=\left(\sum_{i=1}^M w_i\right)\times\ee\left|Y_\beta-\beta\right|
\quad \text{and} \quad P^M_\delta=\left(\sum_{i=1}^M w_i\right)\times\ee\left|Y_\delta-\delta\right|.\]
Consequently, the optimal choice of the parameters is the
median\footnote{The median of a real random variable $X$ is any real
number $m$ satisfying : \[\pp(X \leq m) \geq \frac{1}{2} \,\, \text{
and }\,\,\pp(X \geq m) \geq \frac{1}{2}.\] It has the property of
minimizing the $L^1$-distance to $X$ : $\ds m=\arg\min_{x\in\rr}
\ee|X-x|.$} of $Y_\beta$ for $\beta$ and the median of $Y_\delta$
for $\delta$. Nevertheless, to preserve the interpretation of $\beta_j$
as $\frac{Cov(r_j,r_I)}{Var(r_I)}$ which is equal to
$\frac{\beta\beta_j}{\beta^2}$ for the
simplified index dynamics, one should take $\beta=1$. In Table \ref{tab:poids}, we see that
on the example of the Eurostoxx index at December 21
2007, the optimal choice of $\beta$ is very close to 1 and that the
quantities of interest, $(P^M_{\beta_{opt}})^2$ and
$(P^M_{\beta=1})^2$ are also very close to each other.

\begin{table}[!h]
\begin{center}
\begin{tabular}{|c|c|c|c|}
\hline
$(P^M_w)^2$ &$\beta_{opt}$&$(P^M_{\beta_{opt}})^2$&$(P^M_{\beta=1})^2$\\
\hline 0.026&0.975&0.0173&0.0174\\ \hline
\end{tabular}
\end{center}
\caption{Computation of $(P^M_w)^2, \beta_{opt}$ and
$(P^M_{\beta_{opt}})^2$ for the Eurostoxx index at December 21,
2007. The beta coefficients are estimated on a two year history.}
\label{tab:poids}
\end{table}

The next theorem states that, under an additional assumption on the
volatility coefficients, the $L^{2p}$-distance between a stock
$(S^{j,M}_{t})_{t\in [0,T]}$ and the solution of the SDE obtained by
replacing $I^M$ by $I$
\begin{equation}
   \frac{dS^{j}_t}{S^{j}_t}=(r-\delta_j) dt + \beta_j
\,\sigma(t,I_t) dB_t + \eta_j(t,S^{j}_t) dW^j_t,\;S^j_0=s^{j}_0\label{edsjlim}
\end{equation}
is also controlled by $2p$-powers of $P^M_w$, $P^M_\beta$ and
$P^M_\delta$. One major drawback of the limiting simplified model
\eqref{Ilim}-\eqref{edsjlim} is that the limit index $I_t$ is only
approximately equal to the reconstructed index level
$\overline{I}^M_t\stackrel{\rm def}{=}\sum_{j=1}^M w_j S^j_t$. The next result also gives
an estimation of the
difference between $I^M$ and $\overline{I}^M$ in terms of $P^M_w$, $P^M_\beta$ and
$P^M_\delta$, which combined with the previous theorem, provides an
estimation of the difference between $I$ and $\overline{I}^M$.
\begin{theorem}\label{constock}
Let $p\in{\mathbb N}^*$. Under the assumptions of Theorem 1 and if
\begin{hypo}\label{hyp:xetaLip} $\exists K_\eta$ such that $\forall j,\;\forall
(t,s_{1},s_{2}) \in [0,T]\times \rr_+\times \rr_+,\quad
|s_{1}\eta_j(t,s_{1})-s_{2}\eta_j(t,s_{2})|\leq K_\eta|s_{1}-s_{2}|$.

$\exists K_{Lip}$ such that $\forall (t,s_{1},s_{2}) \in [0,T]\times
\rr_+\times \rr_+,\quad |\sigma(t,s_{1})-\sigma(t,s_{2})|\leq
K_{Lip}|s_{1}-s_{2}|$.
\end{hypo}
Then, $\forall j \in \{1,\dots, M\}$,
\[\ee\left(\sup_{0\leq t \leq T}|S^{j,M}_t-S^j_t|^{2p}\right) \leq
\widetilde{C}^j_T \left(\left(\sum_{j=1}^{M} w_{j}^2\right)^{\!\!p}
+ \left(\sum_{j=1}^M
  w_j|\beta_j-\beta|\right)^{2p}+ \left(\sum_{j=1}^M
w_j|\delta_j-\delta|\right)^{2p} \right)\] where
\[\widetilde{C}^j_{T}=6^{2p-1}K_pT^p\beta_{j}^{2p}C_{2p}^{\frac{1}{2}} K_{Lip}^{2p} \,\, e^{3^{2p-1}((r-\delta_{j})^{2p}T^{2p-1}+K_pT^{p-1}K_\eta^{2p}+
2^{2p-1}K_pT^{p-1}\beta_{j}^{2p}K_b^{2p})T}.\] Moreover, for
$\overline{I}_t^M=\sum_{j=1}^M w_j S^j_t$, one has
\[\ee\left(\sup_{0\leq t \leq
T}|I^M_t-\overline{I}^{M}_t|^{2p}\right) \leq \widetilde{C}_T
\left(\sum_{j=1}^{M} w_{j}\right)^{2p}\left(\left(\sum_{j=1}^{M}
w_{j}^2\right)^{\!\!p}+ \left(\sum_{j=1}^M
  w_j|\beta_j-\beta|\right)^{2p}+ \left(\sum_{j=1}^M
w_j|\delta_j-\delta|\right)^{2p} \right)\] where $\ds
\widetilde{C}_T=\max_{1\leq j \leq M} \widetilde{C}^j_T$.
\end{theorem}

The proof can also be found in the appendix. In the following
corollary,  we consider the limit ${M \to \infty}$ supposing that
the weight of the $j$-th stock, now denoted by $w_j^M$,
depends on $M$.

\begin{Corollary}
Under the assumptions of Theorems 1 and 2 and if

\begin{hypo}
\label{hyp:CT} there exists  a finite constant $A$ such that $\ds
\max_{j\geq 1} \left((s^j_0)^{2} + (\beta_j)^2 +
(\delta_j)^2\right)\leq A$,
\end{hypo}
\begin{hypo}
\label{hyp:i0} $\ds I^M_0=\sum_{j=1}^Mw_j^Ms^j_0\ov{\longrightarrow}_{M \to \infty} i_0\in(0,+\infty)$,
\end{hypo}
\begin{hypo}
\label{hyp:w} $\ds P^M_w=\sqrt{\sum_{j=1}^M (w_j^M)^2}
\ov{\longrightarrow}_{M \to \infty} 0$,
\end{hypo}
\begin{hypo}
\label{hyp:beta} $\ds P^M_\beta=\sum_{j=1}^M w_j^M |\beta_j-\beta|
 \ov{\longrightarrow}_{M \to \infty} 0$,
\end{hypo}
\begin{hypo}
\label{hyp:delta} $\ds P^M_\delta =\sum_{j=1}^M w_j^M
|\delta_j-\delta| \ov{\longrightarrow}_{M \to \infty} 0$,
\end{hypo}
then, for any $p\in\nn^*$, one has
\[\ee\left(\sup_{0\leq t \leq T}|I^{M}_t-I_t|^{2p}\right) \ov{\longrightarrow}_{M \to \infty}
0\;\mbox{ and }\forall j\in\nn^*,\quad \ee\left(\sup_{0\leq t \leq T}|S^{j,M}_t-S^j_t|^{2p}\right) \ov{\longrightarrow}_{M \to \infty}
0.\]
If, in addition, $\ds \sup_M \sum_{j=1}^M w_j^M < \infty$, then
$\ee\left(\sup_{0\leq t \leq T}|I^{M}_t-\overline{I}^M_t|^{2p}\right) \ov{\longrightarrow}_{M \to \infty}
0$.
\end{Corollary}
Assumptions ($\cH$\ref{hyp:i0}), ($\cH$\ref{hyp:w}), ($\cH$\ref{hyp:beta})
and ($\cH$\ref{hyp:delta}) hold for instance when $w^M_j=\frac{1}{M}$
for $1\leq j\leq M$ and $s^j_0\ov{\longrightarrow}_{j\to \infty}
i_0$, $\beta_j\ov{\longrightarrow}_{j\to \infty}
\beta$ and $\delta_j\ov{\longrightarrow}_{j\to \infty}
\delta$.

\subsection*{Simplified model}
To sum up, we have shown that, under mild assumptions, when the
number of underlying stocks is large, the original model may be
approximated by the following dynamics
\begin{equation}
\begin{array}{ll}
\ds \forall j \in \{1,\dots,M\}, &\ds
\frac{dS^{j}_t}{S^{j}_t}=(r-\delta_j) dt + \beta_j
\,\sigma(t,I_t) dB_t + \eta_j(t,S^{j}_t) dW^j_t\\[5mm]
&\ds \frac{dI_t}{I_t}= (r-\delta_I)dt + \sigma(t,I_t) dB_t.
\end{array}
\end{equation}
Of course the distance between this limiting model and the original one
increases with the maturity.

Interestingly, we end up with a local volatility model for the index
and, for each stock, a stochastic volatility model decomposed into a
systemic part driven by the index level and an intrinsic part. The
calibration procedures presented in the next section are based on this
intuition :
even in the original model, we are going to calibrate $\sigma$ as if it
was the local volatility function of the index.

Note
that this simplified model is not valid for options written on the
index together with all its composing stocks since the index is no
longer an exact, but an approximate, weighted sum of the stocks. In
this case, one should consider the reconstructed index
$\overline{I}_t^M=\sum_{j=1}^M w_j S^j_t$ or use the original model.
The simplified model can be used for options
written on the stocks or on the index or even on the index together
with few stocks.

\section{Model calibration}
Calibration, which is how to determine the model parameters in order
to fit market prices at best, is of paramount importance in
practice. In the following, we try to tackle this issue for both our
simplified and original models.
\subsection{Simplified model}
In the simplified limiting model, the only factor which influences
the dynamics of a given stock is the simplified index $I_t$ which
evolves according to an autonomous SDE. So it is enough to address
the calibration of a given stock together with the index and we drop
the index $j$ of the stock for notational simplicity.
\begin{equation}
\begin{array}{ll}
&\ds\frac{dS_t}{S_t}=(r-\delta) dt + \beta
\,\sigma(t,I_t) dB_t + \eta(t,S_t) dW_t,\;S_0=s_0\\[5mm]
&\ds \frac{dI_t}{I_t}= (r-\delta_I)dt + \sigma(t,I_t) dB_t,\;I_0=i_0.
\end{array}
\label{modelsimp1stock}
\end{equation}
The short interest rate and the dividend yields can be extracted
from the market. The calibration of the local volatility function $\sigma$ to
fit index option prices is a classical problem. According to Dupire
\cite{dupire}, if $C_I(t,K)$ denotes the market price of the call option
with maturity $t$ and strike $K$ written on the index, then for
$$\sigma^2(t,K)=2\frac{\frac{\partial C_I}{\partial
  t}(t,K)+(r-\delta_I)K\frac{\partial C_I}{\partial K}(t,K)+\delta_I
C_I(t,K)}{K^2\frac{\partial^2 C_I}{\partial K^2}(t,K)},$$
one has $C_I(t,K)=\ee\left(e^{-rt}(I_t-K)^+\right)$ for all $t,K>0$. Of
course, in practice the market quotes call options only for a finite
number of couples $(t,K)$. What seems a common practice among banks is
to look for $\sigma$ in a parametric family of functions and compute the
parameters minimizing the distance between these quoted prices and the
call prices associated with the parametrized local volatility
function. Since each practitioner may choose his favorite procedure to
address this classical problem of local volatility calibration, we will
not enter in more details.

We also assume that a local volatility function is associated with the
stock by the same procedure and denote by $v_{loc}(t,x)$ the local variance
function of the stock computed as the square of this local volatility
function. So the local volatility model
\begin{equation}
   \ds\frac{d\overline{S}_t}{\overline{S}_t}=(r-\delta) dt +
\sqrt{v_{loc}(t,\overline{S}_t)}dW_t,\;\overline{S}_0=s_0\label{lvms}
\end{equation}
is calibrated to the quoted prices of vanilla options written on the stock.
In \eqref{modelsimp1stock}, by independence between $B$ and
$W$, the variance of the stock at time
$t$ is equal to $\beta^2\sigma^2(t,I_t)+\eta^2(t,S_t)$. According to
Gyöngy \cite{Gyongy}, if
\begin{equation*}
   \forall
t,x>0,\;\ee\left(\beta^2\sigma^2(t,I_t)+\eta^2(t,S_t)|S_t=x\right)=v_{loc}(t,x)
\end{equation*}
then \eqref{modelsimp1stock} and the local volatility model \eqref{lvms}
induce the same marginal distributions for the stock and therefore the
same prices for the vanilla call options written on it :
$\ee\left(e^{-rt}(S_t-K)^+\right)=\ee
\left(e^{-rt}(\overline{S}_t-K)^+\right)$
for all $t,K>0$. Hence if
\begin{equation}
   \forall
t,x>0,\;\eta(t,x)=\sqrt{v_{loc}(t,x)-\beta^2\ee\left(\sigma^2(t,I_t)|S_t=x\right)},\label{eqvar}\end{equation}
then the stock dynamics in \eqref{modelsimp1stock} is calibrated to the
quoted prices of the vanilla options written on the stock. It remains to
choose the coefficient $\beta$ and the function $\eta$ so that this
equality is satisfied. The fact that the law of $(S_t,I_t)$ given by
\eqref{modelsimp1stock} and therefore
the conditional expectation in \eqref{eqvar} depend on $(\beta,\eta)$
makes this problem difficult. Nevertheless, intuitively, when one fixes
a value of $\beta$ that is not too large,
one should be able to find a function
$\eta$ such that \eqref{eqvar} is satisfied.  The calibration
of the stock smile seems over-parametrized and one should rely on the
interpretation of $\beta$ as a regression coefficient to choose its
value. This issue is discussed in the next section. Then we explain how to approximate the
conditional expectation and deduce $\eta$ for a fixed value of $\beta$.\\Let us already point out that
the calibration of our
simplified model gives an advantage to the fit of index option prices in
comparison with options written on the stocks, which is in line with
the market since index options are usually very liquid in comparison
with individual stock options.
\subsubsection{Choice of the coefficient $\beta$}\label{choixbeta}
The interpretation of $\beta$ as the regression coefficient of the
log-returns of the stock with respect to the log-returns of the index
makes it possible to estimate this coefficient on historical data.
Nevertheless, when the historical estimator
$\beta_{hist}$ is large, then the difference in the r.h.s. of
\eqref{eqvar} may become negative for some $(t,x)$ when
$\beta=\beta_{hist}$. Then the square root is no longer defined and
calibration for this choice of $\beta$ is no longer possible.

In Figure~\ref{fig:vollocproblems}, we have plotted the local
volatility of the stock $x\mapsto\sqrt{v_{loc}(T,s_0x)}$, the local
volatility of the index $x\mapsto\sigma(T,i_0x)$, the systemic part
of the volatility of the stock $x\mapsto\beta_{hist} \sigma(T,s_0x)$
and $x\mapsto\beta_{hist}
\sqrt{\ee\left(\sigma^2(T,I_T)|S_T=s_0x\right)}$ when $\eta$ is set
to zero (which intuitively gives the lowest local volatility
function of the stock that one can obtain in our model
\eqref{modelsimp1stock}) as functions of the moneyness for a
maturity $T=1$ year. We considered three representative components
of the Eurostoxx which is composed of $M=50$ stocks : AXA, ALCATEL
and CARREFOUR at December 21, 2007. We made this choice deliberately
in order to point out the extreme situations that one can face :
\begin{itemize}
\item AXA is an example of a stock with a high historical beta coefficient ($\beta_{hist}=1.4$),
\item CARREFOUR is an example of a stock with a low historical beta coefficient ($\beta_{hist}=0.7$),
\item ALCATEL is an example of a stock with a high volatility level but
  with a rather flat smile ($\beta_{hist}=1.1$).
\end{itemize}

Clearly, we can deduce that the market is choosing a $\beta$
coefficient for both AXA and ALCATEL that is lower than the
historical one whereas, for CARREFOUR, one can plug the historical
$\beta$, or even a larger one, in~(\ref{modelsimp1stock}) and still
be able to calibrate the model.

\newpage
\begin{figure}
\begin{center}
\psfrag{bbbbbbbbbbbbbbbbbbbbbbbbbbb1}{\tiny{$\,\sqrt{v_{loc}(T,s_0x)}$}}
\psfrag{bbbbbbbbbbbbbbbbbbbbbbbbbbb2}{\tiny{$\,\sigma(T,i_0x)$}}
\psfrag{bbbbbbbbbbbbbbbbbbbbbbbbbbb3}{\tiny{$\,\beta
\sigma(T,i_0x)$}}
\psfrag{bbbbbbbbbbbbbbbbbbbbbbbbbbb4}{\tiny{$\,\beta
\sqrt{\ee(\sigma^2(T,I_T)|S_T=s_0x)}$}} \leavevmode
\epsfig{file=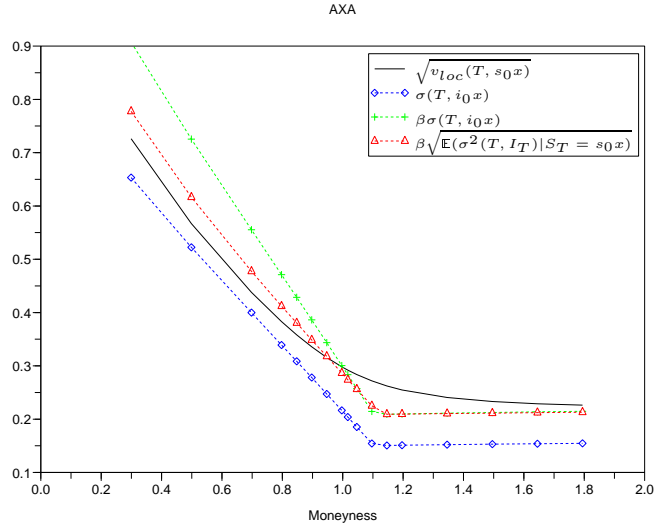,height=50ex}
\epsfig{file=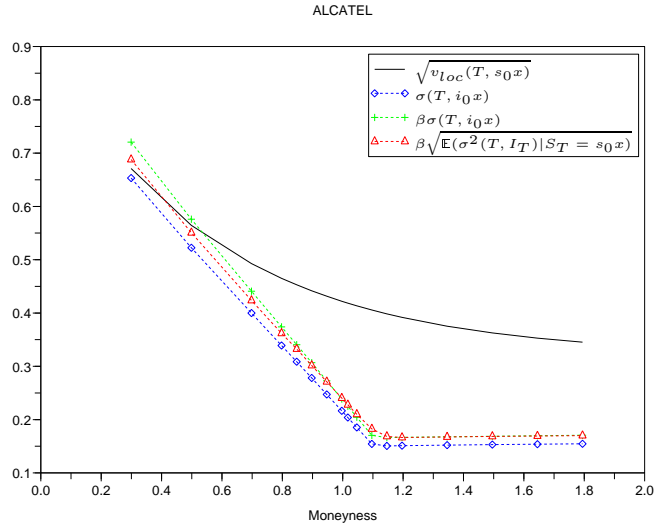,height=50ex}
\epsfig{file=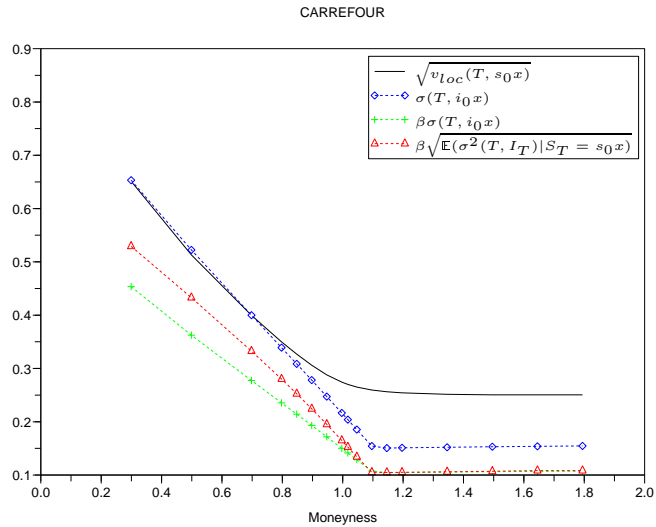,height=50ex}
\end{center}
\caption{Local volatility $x\mapsto\sqrt{v_{loc}(T,s_0x)}$ together
with $x\mapsto\sigma(T,i_0x)$, $x\mapsto\beta_{hist} \sigma(T,i_0x)$
and $x\mapsto\beta_{hist}
\sqrt{\ee\left(\sigma^2(T,I_T)|S_T=s_0x\right)}$ when $\eta$ is set
to zero for AXA, ALCATEL and CARREFOUR.} \label{fig:vollocproblems}
\end{figure}
\thispagestyle{empty}

\clearpage

A satisfactory way to handle the estimation of
the beta coefficient would be to compute an implied beta calibrated
to the prices of options involving the correlation between the stock and the
index. Unfortunately, no such option is liquid in the market (the most liquid
correlation swaps are sensitive to an average correlation
between all the stocks composing the index).

So we suggest to choose
\begin{equation}
   \beta=\min\left(\beta_{hist},\inf_{t,x>0}\frac{\sqrt{v_{loc}(t,s_0x)}}{\sigma(t,i_0x)}\right)
\label{betmin}.\end{equation} Even if we have no proof that this
choice of beta makes the calibration possible, it is sensible and we
have checked that it works on the three examples of AXA, ALCATEL and
CARREFOUR.

When one is interested in options written on
the index together with all its components, one should use the
reconstructed index level $\overline{I}^M_t=\sum_{j=1}^M w_j S^j_t$
instead of $I_t$. Of course, the reconstructed
index dynamics will reproduce the quoted prices of vanilla
options written on the index all the better as $\overline{I}^M_t$ is
close to the calibrated limiting index level $I_t$. According to
Theorems \ref{convind} and \ref{constock}, for this
latter property to hold, one needs
$P^M_{\beta=1}=\sum_{j=1}^M w_j|\beta_j-1|\geq|\sum_{j=1}^Mw_j\beta_j-\sum_{j=1}^M w_j|$
to be small. When, because of the minimum in equality
\eqref{betmin}, $P^M_{\beta=1}$ is larger for the actual
choice of coefficients $\beta$ than for the historical choice, one may
take larger values of beta for stocks like CARREFOUR to decrease
$P^M_{\beta=1}$ and improve the calibration of the
reconstructed index.
\subsubsection{Estimation of the conditional expectation} The idea
behind the following techniques is to circumvent the difficulty of
calibrating the volatility coefficient $\eta$. Indeed, if we plug
the formula \eqref{eqvar} in \eqref{modelsimp1stock}, we
obtain a stochastic differential equation that is nonlinear in the
sense of McKean :

\begin{equation}
\begin{array}{l}
\ds \frac{dS_t}{S_t}=(r-\delta) dt + \beta \,\sigma(t,I_t) dB_t +
\sqrt{v_{loc}(t,S_t)-\beta^2\ee\left(\sigma^2(t,I_t) \,|\,
S_t\right)} dW_t,\;S_0=s_0\\[5mm]
\ds \frac{dI_t}{I_t}= (r-\delta_I)dt + \sigma(t,I_t) dB_t,\;I_0=i_0.
\end{array}
\label{modelsimpNP}
\end{equation}

For an introduction to the stochastic differential equations
nonlinear in the sense of McKean and to propagation of chaos, we
refer to the lecture notes of Sznitman~\cite{Sznitman} and
Méléard~\cite{Meleard}. In our case, the nonlinearity appears in the
diffusion coefficient through the conditional expectation term. This
makes the natural question of existence and uniqueness of a solution
very difficult to handle. The case of a drift coefficient involving
a conditional expectation has only been handled recently even for a
constant diffusion coefficient (see for instance Talay and Vaillant
\cite{TalayVaillant} and Dermoune \cite{Dermoune}). Meanwhile, it is
possible to simulate such a stochastic differential equation by
means of a system of $N$ interacting paths using either a
non-parametric estimation of the conditional expectation or
regression techniques. The advantage of the regression approach over
the non-parametric estimation is that it also yields a smooth
approximation of the function $x\mapsto\ee\left(\sigma^2(t,I_t)
\,|\, S_t=x\right)$ whereas, with a non-parametric method, one has
to interpolate the estimated function and to carefully tune the
window parameter to obtain a smooth approximation.

$\,$\\
\textbf{3.1.2a $\,$ Non-parametric estimation}\\
$\,$

Non-parametric estimators of the conditional expectation, and more
generally non-parametric density estimators, have been widely
studied in the literature. We will focus on kernel estimators of the
Nadaraya-Watson type (see \cite{Watson} and \cite{nadaraya}) : given
$N$ observations $(S_{i,t},I_{i,t})_{i=1\dots N}$ of $(S_t,I_t)$, we
consider the kernel conditional expectation estimator of
$\ee\left(\sigma^2(t,I_t) \,|\, S_t=x\right)$ given by

\[\frac{\ds \sum_{i=1}^N
\sigma^2(t,I_{i,t}) K\left(\frac{x-S_{i,t}}{h_N}\right)}{\ds
\sum_{i=1}^N K\left(\frac{x-S_{i,t}}{h_N}\right)}\] where $K$ is a
non-negative kernel such that $\int_{\mathbb R} K(x)dx=1$ and $h_N$ is a smoothing parameter which tends to
zero as $N\rightarrow +\infty$. This leads to the following system
with $N$ interacting particles : $\forall \, 1\leq i\leq N,$
\begin{equation}
\begin{cases}
    \frac{dS_{i,N,t}}{S_{i,N,t}}=(r-\delta) dt + \beta
\,\sigma(t,I_{i,t}) dB_{i,t} +
\sqrt{v_{loc}(t,S_{i,N,t})-\beta^2\frac{\sum_{j=1}^N
\sigma^2(t,I_{j,t}) K\left(\frac{S_{i,N,t}-S_{j,N,t}}{h_N}\right)}{
\sum_{j=1}^N
K\left(\frac{S_{i,N,t}-S_{j,N,t}}{h_N}\right)}}dW_{i,t},\;S_{i,N,0}=s_0\\[3mm]
 \frac{dI_{i,t}}{I_{i,t}}= (r-\delta_I)dt + \sigma(t,I_{i,t})
dB_{i,t},\;I_{i,0}=i_0
\end{cases}\label{systpart}\end{equation}
where $(B_i,W_i)_{i\geq 1}$ is a sequence of independent
two-dimensional Brownian motions. The integer $i$ indexes the
sample-paths of the fixed stock that we consider. In their dynamics, the
conditional expectation term has been replaced by interaction. The price
in the calibrated model of a European option with maturity $T$ and payoff function $h:C([0,T],\rr)\to\rr$ written on the
stock may be approximated by
\begin{equation}
  \frac{1}{N}\sum_{i=1}^Ne^{-rT}h(S_{i,N,.})
\label{prixcal}.\end{equation}

The $2N$-dimensional SDE may be
discretized using the Euler scheme. Let $n\in\nn^*$ and $0=t_0< \cdots <t_n=T$ be the subdivision with step $\frac{T}{n}$
of $[0,T]$. For each $k\in \{0,\dots,n-1\}$, $\forall \, 1\leq i\leq
N,$
\begin{equation*}
\begin{cases}
\overline{S}_{i,N,t_{k+1}}=\overline{S}_{i,N,t_k}\bigg(1+\sqrt{v_{loc}(t_k,\overline{S}_{i,N,t_k})-\beta^2\frac{\sum_{j=1}^N
\sigma^2(t_k,\overline{I}_{j,t_k})
K\left(\frac{\overline{S}_{i,N,t_k}-\overline{S}_{j,N,t_k}}{h_N}\right)}{\sum_{j=1}^N
K\left(\frac{\overline{S}_{i,N,t_k}-\overline{S}_{j,N,t_k}}{h_N}\right)}} \sqrt{\frac{T}{n}} \tilde{G}_{i,k}
\\\phantom{\overline{S}_{i,N,t_{k+1}}=\overline{S}_{i,N,t_k}\bigg(}+(r-\delta)
\frac{T}{n} + \beta \,\sigma(t_k,\overline{I}_{i,t_k})
\sqrt{\frac{T}{n}} G_{i,k}\bigg) \\[3mm]
\overline{I}_{i,t_{k+1}}=
\overline{I}_{i,t_{k}}\left(1+(r-\delta_I)\frac{T}{n} +
\sigma(t_k,\overline{I}_{i,t_k}) \sqrt{\frac{T}{n}} G_{i,k}\right)
\end{cases}\end{equation*}
where $(G_{i,k})_{1\leq i \leq N,0 \leq k \leq n-1}$ and
$(\tilde{G}_{i,k})_{1\leq i \leq N,0 \leq k \leq n-1}$ are independent
centered and reduced Gaussian random variables.

$\,$\\
\textbf{3.1.2b $\,$ Parametric estimation}\\
$\,$

Another approach to estimate conditional expectations is to use
parametric estimators, or projection. This idea has also been widely
used and studied previously (for example in finance, one can think
of the Longstaff-Schwartz algorithm for pricing American options
\cite{LongstaffSchwartz}). Noting that the conditional expectation
is a projection operator on the space of square integrable random
variables, one can approximate $\ee\left(\sigma^2(t,I_t) \,|\,
S_t=x\right)$ by the parametric estimator $\sum_{l=1}^L \alpha_l f_l(x)$
where $(f_l)_{l=1\dots L}$ is a functional basis and
$\alpha=(\alpha_l)_{l=1\dots L}$ is a vector of parameters estimated
by least mean squares : given $N$ observations
$(S_{i,t},I_{i,t})_{i=1\dots N}$ of $(S_t,I_t)$, $\alpha$ minimizes
$\sum_{i=1}^N \left(\sigma^2(t,I_{i,t}) -\sum_{l=1}^L \alpha_l
f_l(S_{i,t})\right)^2$.

\subsubsection{Numerical results}
\subsubsection*{3.1.3a $\,$ A toy example}

We try to calibrate a stock with a
local variance function $v_{loc}$ constant and equal to $v$. We choose $\sigma$ as the
local volatility function of the Eurostoxx index fitted to the market at
December 21, 2007.

We simulate the system of $N$ interacting paths  \eqref{systpart} and
price call options for different strikes using \eqref{prixcal}. In Figure \ref{fig:volimp2}, we plot the
implied volatility at $T=1$ obtained for independent simulations of $N=5000$
paths and see that they are indeed
close to the desired volatility level $\sqrt{v}$. This example was generated with
the following arbitrary set of parameters :
$$S_0=100,\;\beta=0.7,\;r=0.05,\;\delta=\delta_I=0,\;\sqrt{v}=0.6,\;N=5000,\;n=20.$$

In this example and for all the following numerical experiments, we
use a Gaussian kernel : $K(u)=\frac{1}{\sqrt{2 \pi}}
e^{-\frac{u^2}{2}}$. The smoothing parameter $h_N$ is set to
$N^{-\frac{1}{5}}$ which is the optimal bandwidth that one obtains
when minimizing the asymptotic mean square error of the
Nadaraya-Watson estimator under some regularity assumptions and
assuming independence of the random variables involved (see for
example Bosq~\cite{Bosq}).

\begin{figure}[!ht]
\begin{center}
\leavevmode \epsfig{file=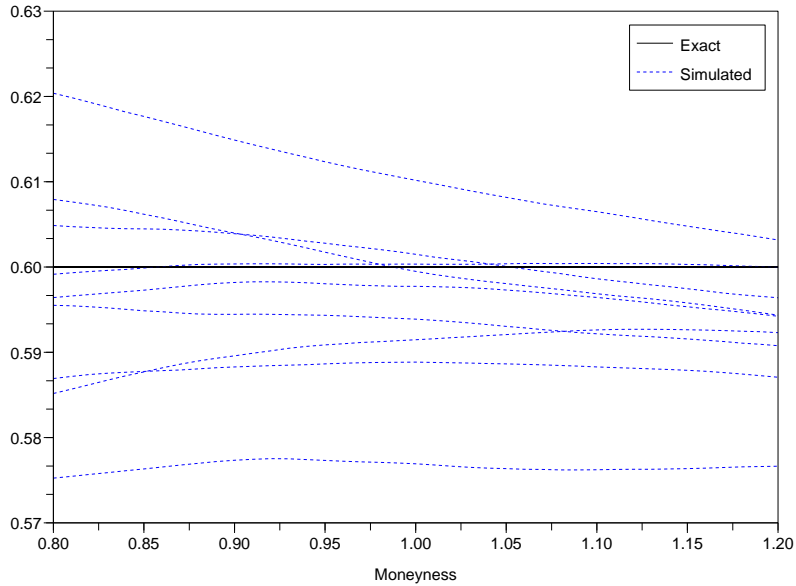,height=60ex}
\caption{Implied volatility obtained for nine independent
simulations with $N=5000$ paths.} \label{fig:volimp2}
\end{center}
\end{figure}

$\,$\\
\textbf{3.1.3b $\,$ An example with real data}

In the following, we test our model with real data. More precisely,
given the local volatilities of the Eurostoxx index and of Carrefour
at December 21, 2007, we simulate (\ref{modelsimpNP}) by different methods for a one year
maturity.

$\,$\\
\textbf{An acceleration technique}

\vspace{2mm} The simulation of the particle system is very time
consuming : for each discretization step and for each stock
particle, one has to make $N$ computations which yield a global
complexity of order $O(nN^2)$ where $n$ is the number of time steps
in the Euler scheme. Acceleration techniques are thus desirable.
One possible method consists in reducing the number of interactions
: instead of making $N$ computations for each estimation of the
conditional expectation, one can neglect interactions which involve
particles which are far away from each other. When the kernel used
is non increasing with the absolute value of its argument, the
easiest way to implement this idea is to sort the particles at each
step and, whenever a contribution of a particle is lower than some
fixed threshold, to stop the estimation of the conditional
expectation.

Of course, by doing this, we lose in precision for the same number
of interacting particles, especially for deep in/out of the money
strikes. But what we gain in terms of computation time is much more
important : in Figure~\ref{fig:volimp3}, we plot the implied
volatility obtained by the naive method and the method with the
above acceleration technique for the same number $N=10000$ of
particles. We take as threshold $\frac{1}{N}$ and set
$h_N=N^{-\frac{1}{10}}$ for the bandwidth parameter\footnote{In
order to smooth the estimation, one has to choose a bandwidth
parameter that is greater than the theoretical optimal parameter
$N^{-\frac{1}{5}}$.} and $n=20$ for the number of time steps in the
Euler scheme. The computation time, on a computer with a 2.8 Ghz
Intel Penthium 4 processor, is of 52 minutes for the naive method
and of 5 minutes for the accelerated one.

\begin{figure}[!ht]
\begin{center}
\leavevmode \epsfig{file=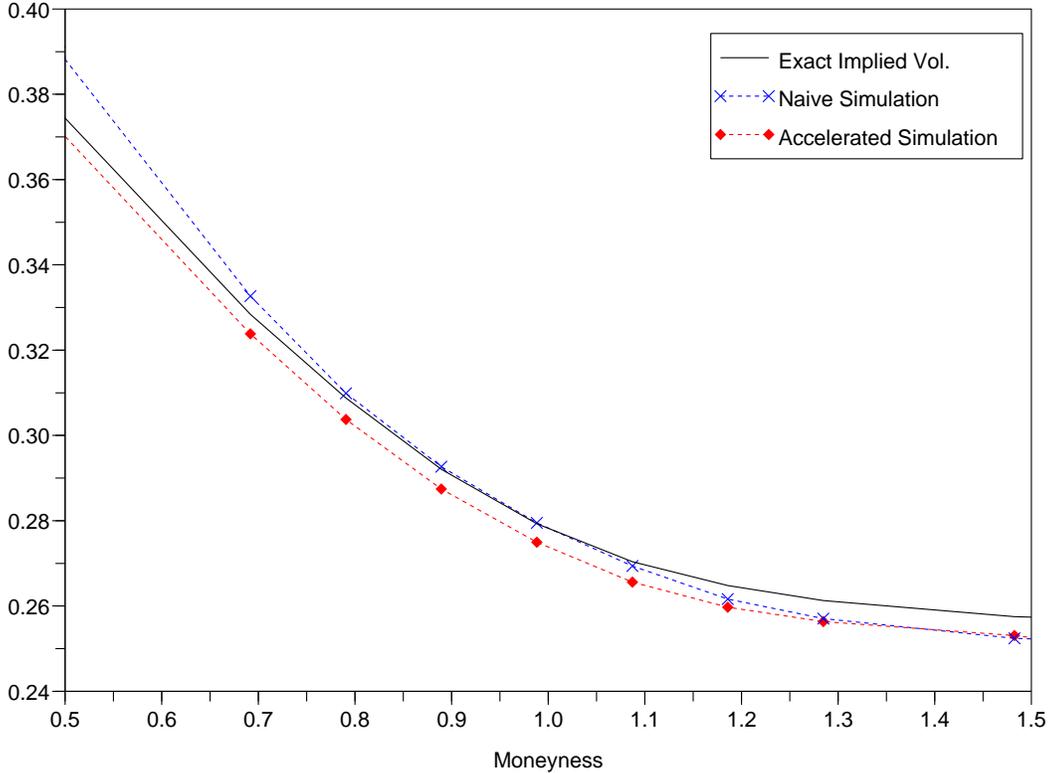,height=80ex}
\caption{Comparison between the naive technique and the accelerated
one for the same number $N=10000$ of particles.} \label{fig:volimp3}
\end{center}
\end{figure}

More importantly, we see that the implied volatility
$\widehat{\sigma}_{N}$ obtained by simulation of the system with $N$
interacting particles converges to the
exact implied volatility $\widehat{\sigma}_{exact}$ computed from quoted option
prices as $N$ tends to $\infty$ : see Figure
\ref{fig:volimp4} and Table \ref{tab:tab1}. With a reasonable number
of simulated paths, $N=200000$, the error on the implied volatility
remains clearly tolerable for practitioners (of the order of 10 bp)
except for a deep in the money call ($K=0.3 S_0$) where it attains
195 bp.

\begin{table}[!ht]
\begin{center}
\begin{tabular}{|c|c|c|c|c|c|c|c|c|c|c|c|}
\hline Moneyness ($\frac{K}{S_0}$) & 0.30 &0.49& 0.69 &0.79& 0.89&
0.99& 1.09& 1.19&
1.28 &1.48 &1.98\\
\hline Error :
$|\widehat{\sigma}_{N}-\widehat{\sigma}_{exact}|$&195& 36& 8& 5
& 2& 1& 2& 9&17&32&56\\
\hline
\end{tabular}
\end{center}
\caption{Error (in bp) on the implied volatility with $N=200000$
particles.} \label{tab:tab1}
\end{table}

\pagebreak

\begin{figure}[!ht]
\begin{center}
\leavevmode \epsfig{file=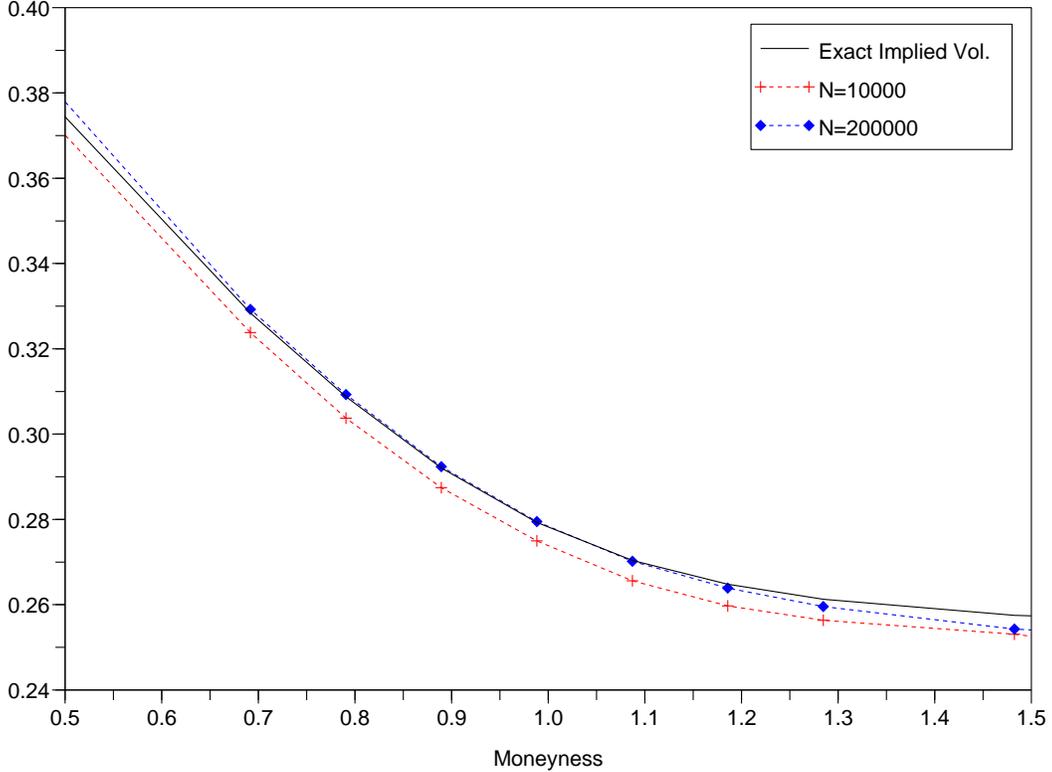,height=80ex}
\caption{Convergence of the implied volatility obtained with
non-parametric estimation as $N\to+\infty$.} \label{fig:volimp4}
\end{center}
\end{figure}

$\,$\\
\textbf{Independent particles}

\vspace{2mm} Unlike the parametric method, non-parametric estimation
of the conditional expectation gives the value of the intrinsic
volatility $\eta$ at the simulated points only. However, using an
interpolation technique, one can first reconstruct $\eta$ with $N_1$
dependent particles and then simulate $N_2$ independent paths of the $2$-dimensional stochastic
differential equation \eqref{modelsimp1stock}. By doing so, we speed up the simulations but one
has to choose carefully the size $N_1$ of the particle system in
order to have a reasonable estimation of the intrinsic volatility
and to tune the bandwidth parameter in order to smooth the
estimation (our numerical tests were done with $N_1=1000,
N_2=100000$ and $h_{N_1}=N_1^{-\frac{1}{10}}$). In
Figures~\ref{fig:volimp5} and~\ref{fig:volimp6}, we plot the local
volatility function $\sqrt{v_{loc}(t,x)}$ and the intrinsic volatility
function $\eta(t,x)$ of the stock. This latter is used to draw independent simulations of
the index along with the stock and we see in
Figure~\ref{fig:volimp7} that the implied volatility obtained is
close to the right one, especially near the money.

\begin{figure}[h!]
   \begin{minipage}[c]{.46\linewidth}
      \includegraphics[scale=0.65]{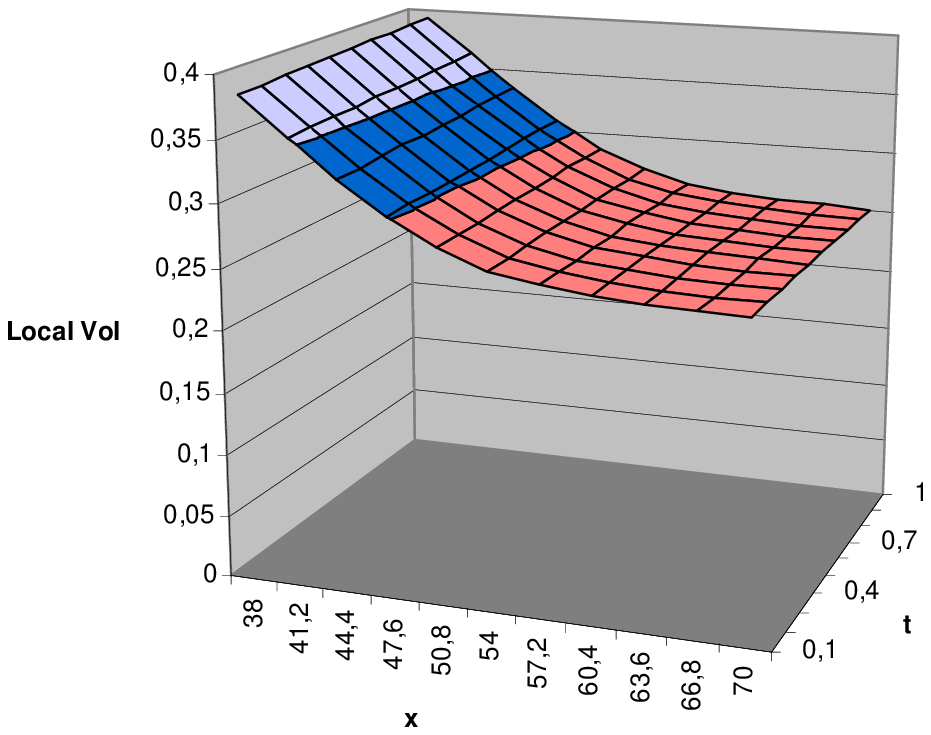}
      \caption{Local volatility function $\sqrt{v_{loc}(t,x)}$ of the stock.}
\label{fig:volimp5}
   \end{minipage} \hfill
   \begin{minipage}[c]{.46\linewidth}
      \includegraphics[scale=0.65]{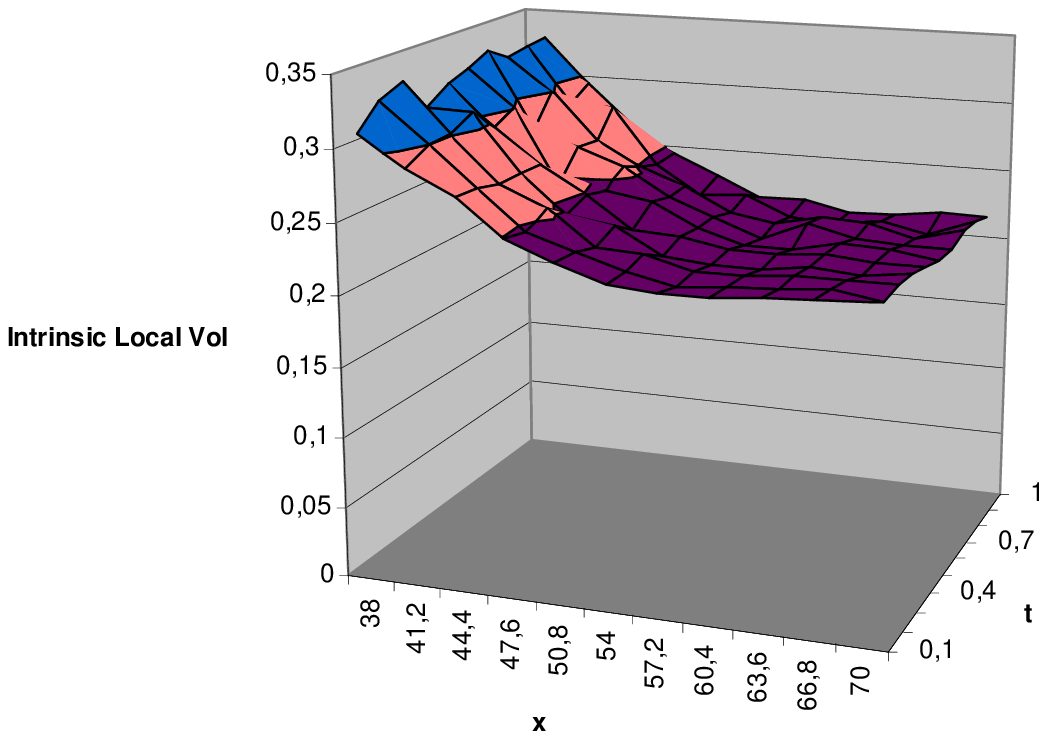}
      \caption{Intrinsic volatility function $\eta(t,x)$ of the stock.}
           \label{fig:volimp6}
   \end{minipage}

\end{figure}

\begin{figure}[h!]
\begin{center}
\leavevmode \epsfig{file=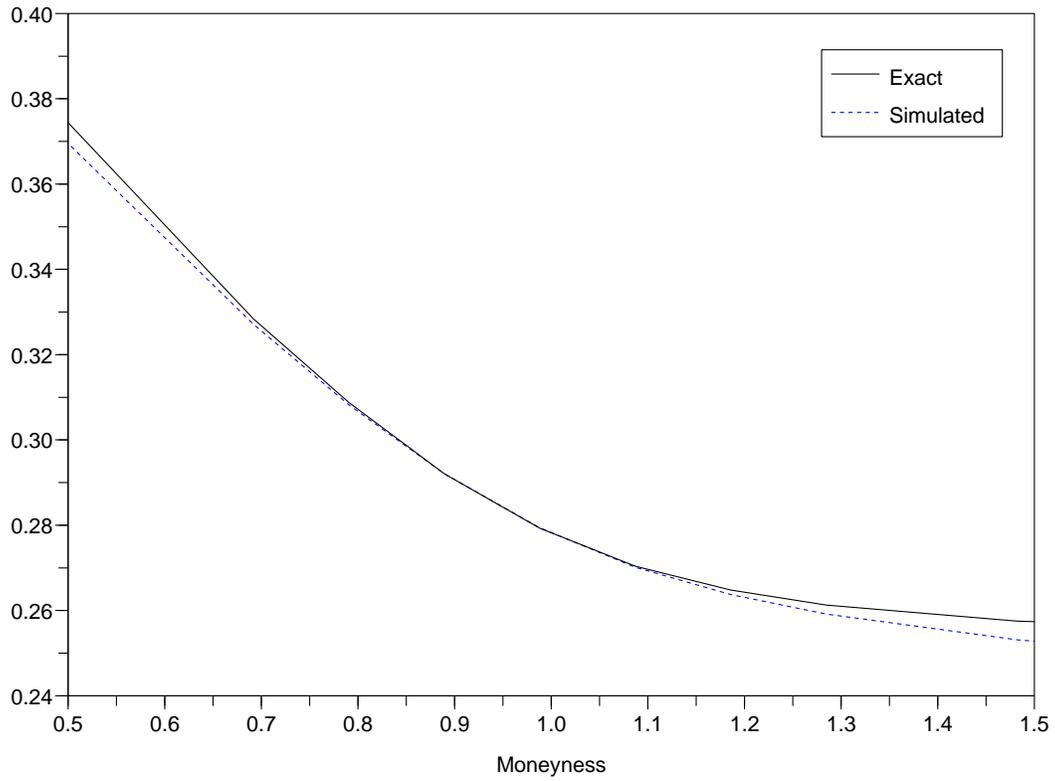,height=80ex}
\caption{Simulated implied volatility with independent draws.}
\label{fig:volimp7}
\end{center}
\end{figure}
\subsection{Original model}
We now turn to the calibration of our original model :
\begin{equation}
\forall j \in \{1,\dots,M\}, \quad
\frac{dS^{j,M}_t}{S^{j,M}_t}=(r-\delta_j) dt + \beta_j
\,\sigma(t,I^{M}_t) dB_t + \eta_j(t,S^{j,M}_t) dW^j_t
\mbox{ with }I^{M}_t=\sum_{i=1}^M w_i S^{i,M}_t.\label{original_model}
\end{equation}

It is rather complicated to have a perfect calibration
for both index and stocks within this framework. Nevertheless, Theorem
\ref{convind} ensures that the error of calibration of the index smile
is small (at least when the maturity is not too large)
when $\sigma$ is chosen as a local volatility function fitted to this
smile.
We also suppose that a
local volatility function $\sqrt{v_{loc}^j}$ has been fitted to the
market smile of each stock $j$. For the choice of the coefficients
$\beta_j$, we proceed like in Section \ref{choixbeta}. The coefficients
$\eta_j(t,x)=\sqrt{v_{loc}^j(t,x)-\beta_j^2\ee(\sigma^2(t,I^M_t)|S^{j,M}_t=x)}$
are then calibrated all at the same time using an adaptation of the
non-parametric method presented above based on the simulation of $N$
interacting $(M+1)$-dimensional paths.

In comparison with the simplified model, we introduce in the
calibration of the index a
small error which grows with the
maturity $T$. But we guarantee the additivity
constraint $I^{M}_t=\sum_{i=1}^M w_i S^{i,M}_t$. Note that a similar
error spoils the calibration of the
reconstructed index in the simplified model (see the discussion at the
end of Section \ref{choixbeta}).

In what follows, we
illustrate the effect of Theorems 1 and 2 and compare our models
with a constant correlation model.

\section{Illustration of Theorems 1 and 2 and comparison with a constant correlation model}
The objective of this section is to compare index and individual
stock smiles obtained with three different models : our original
model (\ref{original_model}), the simplified one (after letting
$M\to \infty$) and a model with constant correlation coefficient.
More precisely, we consider the following dynamics

\begin{enumerate}
\item The original model
\begin{equation}
\begin{array}{l}
\ds \forall j \in \{1,\dots,M\}, \quad
\frac{dS^{j,M}_t}{S^{j,M}_t}=rdt
+ \,\sigma(t,I^{M}_t) dB_t + \eta(t,S^{j,M}_t) dW^j_t\text{ with } I^{M}_t=\sum_{i=1}^M w_i S^{i,M}_t.
\end{array}
\end{equation}
\item The simplified model
\begin{equation}
\begin{array}{ll}
\ds \forall j \in \{1,\dots,M\}, &\ds \frac{dS^{j}_t}{S^{j}_t}=r dt
+ \sigma(t,I_t) dB_t + \eta(t,S^{j}_t) dW^j_t\\[5mm]
&\ds \frac{dI_t}{I_t}= r dt + \sigma(t,I_t) dB_t.
\end{array}
\end{equation}
Here we can also compute the reconstructed index
$\overline{I}^M_t=\sum_{i=1}^M w_i S^{i}_t$.
\item The "market" model
\begin{equation}
\forall j \in \{1,\dots,M\}, \frac{dS^{j}_t}{S^{j}_t}=r dt +
\sqrt{v_{loc}(t,S^j_t)} d\widetilde{W}^j_t
\end{equation}
with, $\forall i \neq j,
\,d\!<\widetilde{W}^i,\widetilde{W}^j>_t=\rho \,dt$.
\end{enumerate}
We deliberately dropped the dividend yields and the beta
coefficients in order to simplify the numerical experiment. For the
function $\sigma$, we take as previously the
calibrated local volatility of the Eurostoxx. For $\eta$, which does not depend on $j$, we choose an arbitrary
function of the forward moneyness and we evaluate $v_{loc}$ such that
the ``market'' model and the simplified model yield the same implied
volatility for individual stocks. According to \cite{Gyongy}, it is enough to take
\[v_{loc}(t,x)=\eta^2(t,x)+\ee(\sigma^2(t,I_t) | S^1_t=x)\]
where the conditional expectation is approximated using the
non-parametric method presented above.

Finally, we fix the correlation coefficient $\rho$ such that the
market model and the simplified one have the same ATM implied
volatility for the index.

The implied volatilities for the index and for an individual stock
obtained by the three models are plotted in Figures
\ref{fig:IndexCompare} and \ref{fig:StockCompare}. We also give the
difference in basis points between the implied volatilities obtained
with the simplified model and the original one in Tables
\ref{tab:theo1}, \ref{tab:theo21} and \ref{tab:theo22}. The
parameters we use in our numerical experiment are the following :
\begin{itemize}
\item[-]$S_0^1=\dots=S_0^M=53$,
\item[-]$M$, $I_0$ and the weights $w_1, \dots, w_M$ : the same as of the Eurostoxx index at December 21, 2007,
\item[-]$r=0.045$,
\item[-]Maturity $T=1$ year,
\item[-]Number of time steps: $n=10$,
\item[-]Number of simulated paths : $N=100000$.
\end{itemize}
\begin{figure}[!ht]
\begin{center}
\leavevmode \epsfig{file=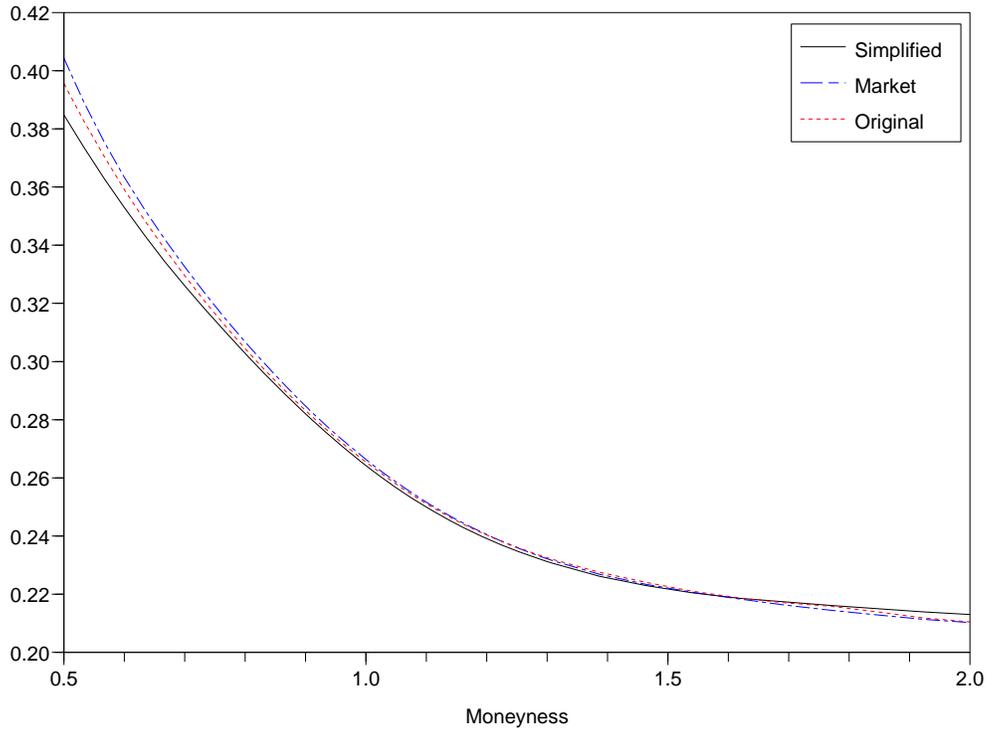,height=75ex} \caption{Implied
volatility of an individual stock.} \label{fig:StockCompare}
\end{center}
\end{figure}

\begin{table}[!ht]
\begin{center}
\begin{tabular}{|c|c|c|c|c|c|c|c|c|c|c|c|c|}
\hline Moneyness ($\frac{K}{S_0}$) & 0.5 &0.8& 0.9 &0.95& 1& 1.05&
1.1& 1.2&
1.3 &1.55 &1.85&2\\
\hline
$|\widehat{\sigma}_{simplified}-\widehat{\sigma}_{original}|$&81
&22& 16 &14 &14 &17 &20& 24& 24& 11& 38& 17\\
\hline
\end{tabular}
\end{center}
\caption{Difference (in bp) the implied
volatilities of an individual stock obtained with the simplified model and with the original model.} \label{tab:theo22}
\end{table}

\begin{figure}[!ht]
\begin{center}
\leavevmode \epsfig{file=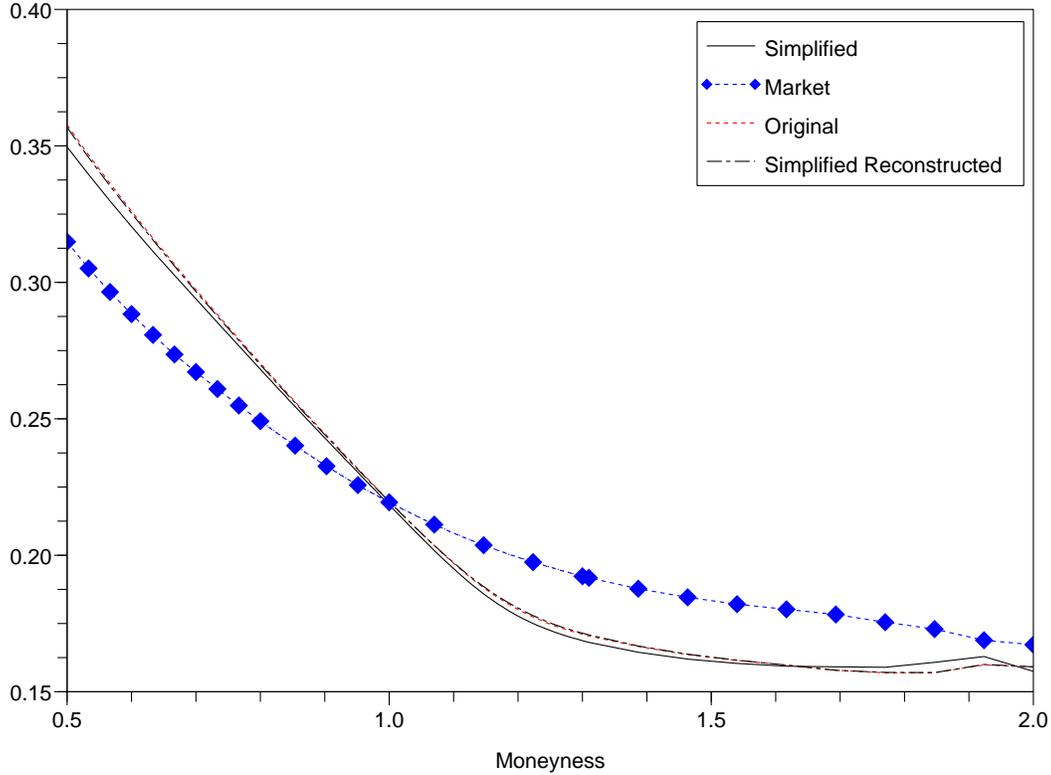,height=80ex} \caption{Implied
volatility of the index.} \label{fig:IndexCompare}
\end{center}
\end{figure}

\begin{table}[!ht]
\begin{center}
\begin{tabular}{|c|c|c|c|c|c|c|c|c|c|c|c|c|}
\hline Moneyness ($\frac{K}{I_0}$) & 0.5 &0.8& 0.9 &0.95& 1& 1.05&
1.1& 1.2&
1.3 &1.55 &1.85&2\\
\hline
$|\widehat{\sigma}_{simplified}-\widehat{\sigma}_{original}|$&81
&22& 16 &14 &14 &17 &20& 24& 24& 11& 38& 17\\
\hline
\end{tabular}
\end{center}
\caption{Difference (in bp) between the implied volatilities of the index
obtained with the simplified model and with the
original model.} \label{tab:theo1}
\end{table}

\begin{table}[!ht]
\begin{center}
\begin{tabular}{|c|c|c|c|c|c|c|c|c|c|c|c|c|}
\hline Moneyness ($\frac{K}{I_0}$) & 0.5 &0.8& 0.9 &0.95& 1& 1.05&
1.1& 1.2&
1.3 &1.55 &1.85&2\\
\hline
$|\widehat{\sigma}_{reconstruct}-\widehat{\sigma}_{original}|$&10&
5&
4& 3& 2& 1& 2& 5& 4 &1& 0& 0\\
\hline
\end{tabular}
\end{center}
\caption{Difference (in bp) between the implied volatility of the
reconstructed index $\overline{I}^M$ in the simplified model and the
implied volatility of the index in the original model.}
\label{tab:theo21}
\end{table}

As suggested by Theorems 1 and 2, we see that the original model and
the simplified one yield implied volatility curves that are very
close to each other, both for the index and for individual stocks.
The difference in basis points between the implied volatilities is
reasonable, especially between the reconstructed index in the simplified model and
the index
in the original model.

Concerning the market model, by construction, we have the same
implied volatility for an individual stock as in the simplified
model but the implied volatility of the index is far from
the simplified one.
This phenomenon is well known in practice (see
\cite{Bakshi},\cite{Bollen} or\cite{Branger}) : the implied
volatility smile of an index is much steeper than the implied
volatility smile of an individual stock. The market model of
constantly correlated local volatility dynamics for the stocks is unable to retrieve the shape of the
index smile. A more sophisticated dependence structure between stocks
is needed. Local correlation
models provide an extension of the market model in this direction : the correlation at time $t$ between the Brownian motions
driving the local volatility dynamics of the stocks is a function
$\rho(t,I_t)$ of the index level. But the way this function $\rho$
influences the index smile is not clear at all. Somehow, our models
provide another parametrization of the correlation structure in which,
the function $\sigma$, that replaces the function $\rho$, can be
interpreted as the local volatility of the index. Yet, the
individual stocks can still be properly calibrated.
$$\;$$\textbf{Application: Pricing of a worst-of option}$$\;$$
Apart from handling both the index and its composing stocks, our
models are also relevant for the widespread financial products that
are sensitive to correlation in the equity world, such as rainbow
options.

One example of such products is the worst-of performance option
whose payout is referenced to the worst performer in a basket of
shares. For a basket of $M$ shares, the payoff of a call with strike
$K$ and maturity $T$ writes $\ds\left(\min_{1\leq i\leq
M}\frac{S_T^i}{S_0^i}-K\right)_+$. Our objective is to compare the
prices obtained by our model to the prices obtained by the market
model of constantly correlated stocks. The parameters of the
numerical experiment are the same as previously and we set the
correlation coefficient $\rho$ such that all the models exhibit the
same ATM implied volatility for the index.

The result, as can be seen in Figure~\ref{fig:worst_of}, is that our
prices are always lower than the market model price, especially in
the money. Hence, a model with a constant correlation coefficient,
calibrated in order to fit the at the money prices of options written on the
index, will always
overestimate the risks of worst-of options. The reason is that the
correlation level needed to fit the at the money prices is very high. Note that the prices
obtained with the original model and the simplified one are barely
distinguishable from each other.

\clearpage

\begin{figure}[h]
\begin{center}
\leavevmode \epsfig{file=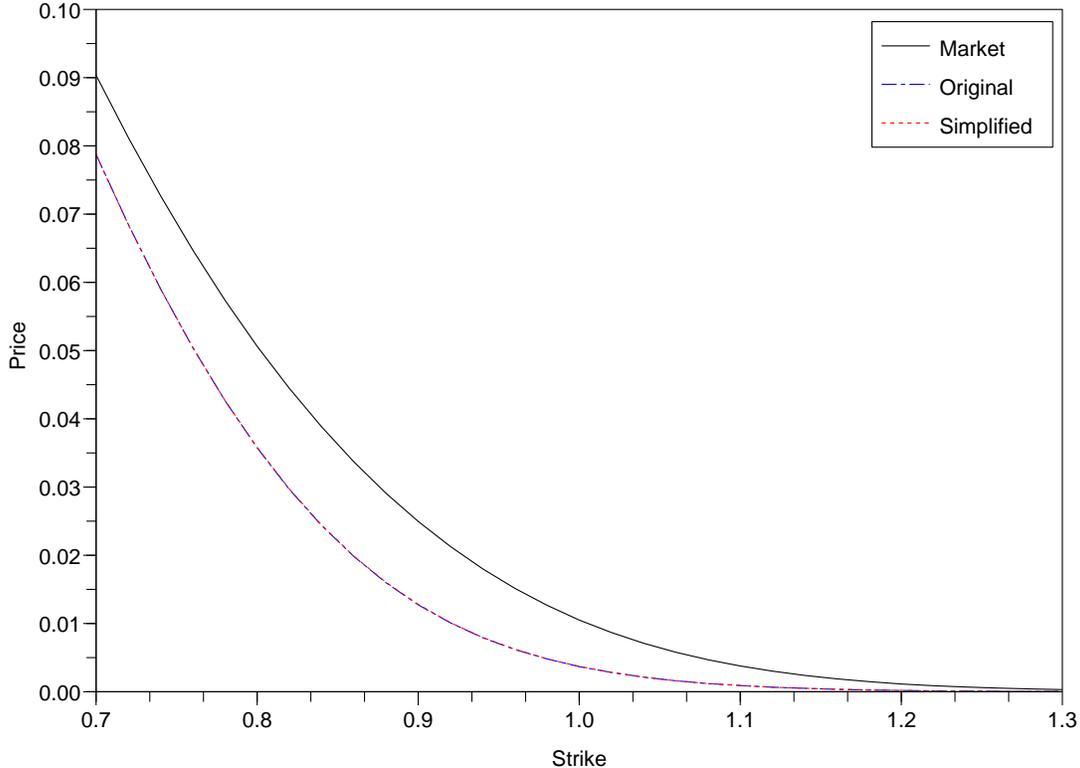,height=80ex} \caption{Worst-of
price.} \label{fig:worst_of}
\end{center}
\end{figure}

\section{Conclusion}
In this paper, we have introduced a new model for describing the
joint evolution of an index and its composing stocks. The idea
behind our view is that an index is not only a weighted sum of
stocks but can also be seen as a market factor that influences their
dynamics. In order to have a more tractable model, we have studied
the limit when the number of underlying stocks goes to infinity and
we have shown that our model reduces to a local volatility model for
the index and to a stochastic volatility model with volatility driven by the index for each individual
stock. We have discussed calibration issues and proposed a
simulation-based technique for the calibration of the stock
dynamics, which permits us to fit both index and stocks smiles. The
numerical results obtained on real data for the Eurostoxx index are
very encouraging, especially for accelerated techniques. We have
also compared our models (before and after passing to the limit) to
a standard market model consisting of local volatility models for
the stocks which are constantly correlated and we have seen that they
lead to a steeper index smile. Finally,
when considering the pricing of worst-of performance options, which
are sensitive to the dependence structure between stocks, we have
found that our prices are more aggressive than the prices obtained
by the standard market model.

To sum up, we list some properties of our models depending on the
options one wishes to handle in the Table below

\begin{table}[!ht]
\begin{center}
\begin{tabular}{|l|l|l|}
\hline  \hspace{8mm} Purpose & \hspace{18mm} Simplified model & \hspace{18mm} Original model \\
\hline \small{Options written on }&
-\small{Simulation of a $(J+1)$-dimensional SDE :}&-\small{Simulation of an $M$-dimensional SDE :}\\
-\small{few ($J<<M$) stocks}&\small{$\,\,(I,S^1,\dots,S^J)$.}&\small{$\,\,(S^{1,M},\dots,S^{M,M})$.}\\
-\small{the index.}&-\small{Exact calibration of $(S^j)_{1\leq j\leq
J}$ and $I$
possible.}&-\small{Exact calibration of $(S^{j,M})_{1\leq j\leq J}$
possible}\\

& &\small{but requires calibration of all the stocks.}\\
&&-\small{Approximate calibration of $I^M$.}\\

\hline \small{Options written on }&
-\small{Simulation of an $(M+1)$-dimensional SDE :}&-\small{Simulation of an $M$-dimensional SDE :}\\
-\small{all the
stocks}&\small{$\,\,(I,S^1,\dots,S^M)$.}&\small{$\,\,(S^{1,M},\dots,S^{M,M})$.}\\
-\small{the index.}&-\small{Exact calibration of all the stocks
possible.}&-\small{Exact calibration of all the stocks
possible.}\\
&-\small{Index value : $\overline{I}^M_t=\sum_{j=1}^M w_j S^j_t$}.&-\small{Approximate calibration of $I^M$.}\\
&-\small{Approximate calibration of $\overline{I}^M$.}&\\ \hline
\end{tabular}
\end{center}
\caption{Which model to use and when.} \label{tab:summary}
\end{table}

\pagebreak

\bibliographystyle{plain}

\pagebreak
\begin{center}
\textbf{\large Appendix}
\end{center}

\vspace{5mm}

In order to prove the Theorems 1 and 2, we need the following
technical estimation

\begin{lemma}
\label{lem:moments} Under assumption ($\cH$\ref{hyp:bornitude}), for
all $p \geq 1$, one has
\begin{equation}
\forall j \in \{1,\dots,M\},\quad  \sup_{0\leq t\leq
T}\ee\left(|S^{j,M}_{t}|^{2p}\right) \leq  C_p
\end{equation}
where $\ds C_p=\max_{1\leq j \leq M}|S^{j,M}_{0}|^{2p}
\exp\left(\left(2r+(2p-1)(\max_{j \geq
1}\beta_j^2+1)K_b^2\right)pT\right)$.
\end{lemma}
\begin{proof}
By Itô's lemma one has
\[\begin{array}{rcl}
\ds|S^{j,M}_{t}|^{2p}&=&\ds|S^{j,M}_{0}|^{2p}+\int_0^t|S^{j,M}_{s}|^{2p}((2p)(r-\delta_j)+p(2p-1)(\beta_j^2\sigma^2(s,I^M_s)+\eta_j^2(s,S^{j,M}_{s})))ds\\[4mm]
&&\ds\quad+\int_0^t(2p)|S^{j,M}_{s}|^{2p}(\beta_j\sigma(s,I^M_s)dB_s+\eta_j(s,S^{j,M}_{s})dW^j_s)\end{array}\]
In order to get rid of the stochastic integral, we use a
localization technique : let $\nu_n$ be the stopping time defined
for each $n\in \nn$ by $\nu_n:=\inf\{t \geq 0; |S^{j,M}_t| \geq
n\}$. Then, using ($\cH$\ref{hyp:bornitude}), one has
\[\begin{array}{rcl}
\ds \ee\left(|S^{j,M}_{t\wedge \nu_n}|^{2p}\right)\!\!&=&\ds
|S^{j,M}_{0}|^{2p}+\ee\left(\int_0^{t\wedge
\nu_n}\!\!|S^{j,M}_{s}|^{2p}
((2p)(r-\delta_j)+p(2p-1)(\beta_j^2\sigma^2(s,I^M_s)+\eta_j^2(s,S^{j,M}_{s}))ds\right)\\[5mm]
&\leq &\ds
|S^{j,M}_{0}|^{2p}+\left((2p)(r-\delta_j)\mathbb{1}_{\{r-\delta_j
\geq0\}}+p(2p-1)(\beta_j^2+1)K_b^2\right)\int_0^{t}\ee\left(|S^{j,M}_{s\wedge
\nu_n}|^{2p}\right)ds
\end{array}\]
So, by Gronwall's lemma and the fact that the dividends are
nonnegative,
\begin{equation}
\forall t\leq T, \ee\left(|S^{j,M}_{t\wedge \nu_n}|^{2p}\right) \leq
|S^{j,M}_{0}|^{2p}
\exp\left(\left(2rp+p(2p-1)(\beta_j^2+1)K_b^2\right)T\right)\end{equation}

Finally, Fatou's lemma permits us to conclude :
\begin{equation}
\sup_{0\leq t\leq T} \ee\left(|S^{j,M}_{t}|^{2p}\right) \leq
|S^{j,M}_{0}|^{2p}
\exp\left(\left(2rp+p(2p-1)(\beta_j^2+1)K_b^2\right)T\right).\end{equation}
\end{proof}

\begin{proofOFT1}
Using the SDEs (\ref{indexSDE}) and (\ref{Ilim}), one has
\[ \begin{array}{rcl}
 \ds  |I_{t}^{M}-I_{t}|^{2p}&=&\ds \big|r \int_{0}^{t}
  \left(I^{M}_{s}-I_{s}\right)ds-\int_{0}^{t} \left(\sum_{j=1}^M \delta_j w_j S^{j,M}_s-\delta I_{s}\right)ds\\[3mm]
  &&\ds +\int_{0}^{t}\left(\sum_{j=1}^M \beta_j w_j S^{j,M}_s\sigma(s,I^{M}_{s})-\beta
  I_{s}\sigma(s,I_{s})\right)dB_{s}+ \sum_{j=1}^{M} w_{j} \int_{0}^{t}S^{j,M}_{s}\eta_{j}(s,S^{j,M}_{s}) dW_{s}^{j}\big|^{2p}\\[7mm]
  &\leq&\ds 4^{2p-1} \left(r^{2p}t^{2p-1}\int_{0}^{t}(I^{M}_{s}-I_{s})^{2p}ds+t^{2p-1}\int_{0}^{t}\left(\sum_{j=1}^M \delta_j w_j S^{j,M}_s-\delta I_{s}\right)^{2p}ds\right.\\[3mm]
  &&\ds +\left.\Big|\!\int_{0}^{t}\left(\sum_{j=1}^M \beta_j w_j S^{j,M}_s\sigma(s,I^{M}_{s})-\beta
  I_{s}\sigma(s,I_{s})\right)
  dB_{s}\Big|^{2p}+ \Big|\sum_{j=1}^{M}
    w_{j}\!\int_{0}^{t}S^{j,M}_{s}\eta_{j}(s,S^{j,M}_{s})
    dW_{s}^{j}\Big|^{2p}\right)
      \end{array}
  \]
Hence, using the Burkholder-Davis-Gundy inequality (see Karatzas and
Shreve~\cite{KaratzasShreve} p. 166), there exists a universal
positive constant $K_p$ such that
\[\ee\left(\sup_{0\leq t \leq T}|I^{M}_t-I_t|^{2p}\right) \leq  4^{2p-1}(a_{M}+b_{M}+c_{M}+d_{M})
  \]
  where
  \begin{itemize}
    \item $\ds a_{M}=r^{2p} \,T^{2p-1} \int_{0}^{T} \ee\big((I^{M}_{s}-I_{s})^{2p}\big)ds$
    \item $\ds b_{M}=T^{2p-1}\int_{0}^{T} \ee\left(\left(\sum_{j=1}^M \delta_j w_j S^{j,M}_s-\delta I_{s}\right)^{\!\!\!2p}\,\right)ds$
    \item $\ds c_{M}=K_p T^{p-1} \int_{0}^{T}\ee\left(\left(\sum_{j=1}^M \beta_j w_j S^{j,M}_s\sigma(s,I^{M}_{s})-\beta
  I_{s}\sigma(s,I_{s})\right)^{\!\!\!2p}\,\right) ds$
    \item $\ds d_{M}=K_p T^{p-1} \int_{0}^{T}\ee\left(\left(\sum_{j=1}^M \left(w_j S^{j,M}_s\eta_j(s,S^{j,M}_{s})\right)^2\right)^{\!\!\!p}\,\right) ds$
  \end{itemize}

The term $a_{M}$ is the easiest one to handle :
\begin{equation}
a_{M} \leq r^{2p} \,T^{2p-1}\int_{0}^{T}\ee\left(\sup_{0 \leq u \leq
    s}|I^{M}_{u}-I_{u}|^{2p}\right) ds.\label{aM}
\end{equation}

Next, using assumption ($\cH$\ref{hyp:bornitude}) for the first
inequality, Hölder's inequality for the second and lemma
\ref{lem:moments} for the third, one gets
\begin{equation}
\begin{array}{rcl}
d_{M}&=&\ds K_p T^{p-1} \int_0^T \sum_{j_1=1}^{M} \cdots
\sum_{j_p=1}^{M} \ee\left(\prod_{k=1}^p w_{j_k}^2
(S^{j_k,M}_{s})^2(\eta_{j_k}(s,S^{j_k,M}_{s}))^2\right) ds\\[5mm]
&\leq &\ds K_p K_b^{2p}T^{p-1} \int_0^T \sum_{j_1=1}^{M} \cdots
\sum_{j_p=1}^{M} (\prod_{k=1}^p w_{j_k}^2) \ee\left(\prod_{k=1}^p
(S^{j_k,M}_{s})^2\right) ds\\[5mm]
&\leq &\ds K_p K_b^{2p}T^{p-1} \int_0^T \sum_{j_1=1}^{M} \cdots
\sum_{j_p=1}^{M} \prod_{k=1}^p w_{j_k}^2 \left(\ee\left(
(S^{j_k,M}_{s})^{2p}\right)\right)^{\frac{1}{p}} ds\\[5mm]
&\leq &\ds K_p K_b^{2p}T^{p} C_p \left(\sum_{j=1}^{M}
w_{j}^2\right)^{\!\!p}
\end{array} \label{dM}
\end{equation}

The same arguments enable us to control the term $b_M$ :

\begin{equation}\begin{array}{rcl} \ds b_M&=&\ds T^{2p-1}\int_{0}^{T} \ee\left(\left(\sum_{j=1}^M \delta_j w_j S^{j,M}_s-\delta I_{s}\right)^{\!\!\!2p}\,\right)ds\\[5mm]
  &\leq& (2T)^{2p-1}\ds \left(\int_{0}^{T} \ee\left(\left(\sum_{j=1}^M \delta_j w_j S^{j,M}_s-\delta I^M_{s}\right)^{\!\!\!2p}\,\right)
   + \ee\left(\left(\delta I^{M}_{s}-\delta I_{s}\right)^{2p}\right)ds\right)\\[5mm]
  &\leq&(2T)^{2p-1} \ds \int_{0}^{T} \ee\left(\left(\sum_{j=1}^M
  (\delta_j-\delta)
w_j S^{j,M}_s\right)^{2p}\right)ds +(2T)^{2p-1}
\delta^{2p}\int_{0}^{T} \ee\left(\sup_{0 \leq u \leq
    s}|I^{M}_{u}-I_{u}|^{2p}\right) ds\\[5mm]
  &\leq&\ds 2^{2p-1} T^{2p}C_p\left(\sum_{j=1}^M w_j|\delta_j-\delta|\right)^{2p}+(2T)^{2p-1}
\delta^{2p}\int_{0}^{T} \ee\left(\sup_{0 \leq u \leq
    s}|I^{M}_{u}-I_{u}|^{2p}\right) ds.
\end{array}\label{bM}\end{equation}

For the remaining term $c_M$, we will also need the Lipschitz
assumption ($\cH$\ref{hyp:xsLip})
\begin{equation}\begin{array}{rcl}
  \ds c_{M} &=&\ds K_p T^{p-1} \int_{0}^{T}\ee\left(\left(\sum_{j=1}^M \beta_j w_j S^{j,M}_s\sigma(s,I^{M}_{s})-\beta
  I_{s}\sigma(s,I_{s})\right)^{\!\!\!2p}\,\right) ds\\[5mm]
&\leq& \ds 2^{2p-1} K_p T^{p-1} \left( \int_{0}^{T}
\ee\left(\left(\sum_{j=1}^M (\beta_j-\beta) w_j
S^{j,M}_s\sigma(s,I^{M}_{s})\right)^{\!\!\!2p}\,\right)
+\ee\left((\beta I^{M}_{s}\sigma(s,I^{M}_{s})-\beta
    I_{s}\sigma(s,I_{s}))^{2p}\right)ds\right)\\[5mm]
  &\leq&\ds 2^{2p-1} K_pT^{p}K_b^{2p}C_p\left(\sum_{j=1}^M
  w_j|\beta_j-\beta|\right)^{2p}+2^{2p-1}K_pT^{p-1}(\beta K_\sigma)^{2p}
\int_{0}^{T} \ee\left(\sup_{0 \leq u \leq
    s}|I^{M}_{u}-I_{u}|^{2p}\right) ds.
\end{array}\label{cM}\end{equation}

So, combining the inequalities (\ref{aM}), (\ref{dM}), (\ref{bM})
and (\ref{cM}), one obtains
\[\begin{array}{rcl}
\ds \ee\left(\sup_{0\leq t \leq T}|I^{M}_t-I_t|^{2p}\right)
  &\leq&\ds C_0\left(\left(\sum_{j=1}^{M}
w_{j}^2\right)^{\!\!p} + \left(\sum_{j=1}^M
  w_j|\beta_j-\beta|\right)^{2p}+ \left(\sum_{j=1}^M
w_j|\delta_j-\delta|\right)^{2p} \right)\\[5mm]
  &&\ds+ C_1 \int_{0}^{T}\ee\left(\sup_{0 \leq u
\leq
    s}|I^{M}_{u}-I_{u}|^{2}\right) ds\end{array}\]
with $C_0=8^{2p-1} T^p (T^p+K_pK_b^{2p})C_p$ and
$C_1=4^{2p-1}(2^{2p-1}K_pT^{p-1}(\beta
  K_\sigma)^{2p}+(2T)^{2p-1}
\delta^{2p}+r^{2p} \,T^{2p-1}).$

Finally, by means of Gronwall's lemma, we conclude that
\[\ee\left(\sup_{0\leq t \leq T}|I^{M}_t-I_t|^{2p}\right) \leq C_T \left(\left(\sum_{j=1}^{M}
w_{j}^2\right)^{\!\!p} + \left(\sum_{j=1}^M
  w_j|\beta_j-\beta|\right)^{2p}+ \left(\sum_{j=1}^M
w_j|\delta_j-\delta|\right)^{2p} \right)\] where
\[C_T=C_0 e^{C_1T}.\]
\end{proofOFT1}

\begin{proofOFT2}
  The proof is similar to the previous one :
  \[\begin{array}{rcl}
    \ds |S^{j,M}_{t}-S^j_t|^{2p}& \leq &\ds 3^{2p-1}
  \left((r-\delta_{j})^{2p}t^{2p-1}\int_{0}^{t}(S^{j,M}_{s}-S^j_s)^{2p}ds+\left|\int_{0}^{t}(S^{j,M}_{s}\eta_{j}(s,S^{j,M}_{s})-S^{j}_{s}\eta_{j}(s,S^{j}_{s}))dW^{j}_{s}\right|^{2p}\right.\\[2mm]
    &&\ds\left.+\beta_{j}^{2p}\left|\int_{0}^{t}(S^{j,M}_{s}\sigma(s,I^{M}_{s})-S^{j}_{s}\sigma(s,I_{s}))dB_{s}\right|^{2p}\right)
  \end{array}\]
  hence, using the Burkholder-Davis-Gundy inequality, there exists a
  constant $K_p$ such that
 \[\begin{array}{rcl}
\ds   \ee\left(\sup_{0\leq t \leq T}|S^{j,M}_{t}-S^j_t|^{2p}\right)&
\leq&\ds 3^{2p-1}
  \left((r-\delta_{j})^{2p}T^{2p-1}\int_{0}^{T}\ee\left(\sup_{0 \leq u \leq
        s}|S^{j,M}_{u}-S^j_u|^{2}\right)ds\right.\\[2mm]
&&\ds
+K_pT^{p-1}\int_{0}^{T}\ee\left((S^{j,M}_{s}\eta_{j}(s,S^{j,M}_{s})-S^{j}_{s}\eta_{j}(s,S^{j}_{s}))^{2p}\right)ds\\[2mm]
&&\ds
  \left.+ K_pT^{p-1}\beta_{j}^{2p}\int_{0}^{T}\ee\left((S^{j,M}_{s}\sigma(s,I^{M}_{s})-S^{j}_{s}\sigma(s,I_{s}))^{2p}\right)ds\right)
  \end{array}\]
Using assumption ($\cH$\ref{hyp:xetaLip}), one gets
\[\int_{0}^{T}\ee\left((S^{j,M}_{s}\eta_{j}(s,S^{j,M}_{s})-S^{j}_{s}\eta_{j}(s,S^{j}_{s}))^{2p}\right)ds
\leq K_{\eta}^{2p} \int_{0}^{T}\ee\left(\sup_{0 \leq u \leq
        s}|S^{j,M}_{u}-S^j_u|^{2p}\right)ds.\]
Finally, by means of lemma \ref{lem:moments} and assumptions
($\cH$\ref{hyp:bornitude}) and ($\cH$\ref{hyp:xsLip}),
\[\begin{array}{rcl}
  \ds
  \int_{0}^{T}\ee\left((S^{j,M}_{s}\sigma(s,I^{M}_{s})-S^{j}_{s}\sigma(s,I_{s}))^{2p}\right)ds&\leq& \ds 2^{2p-1}
  \int_{0}^{T}\ee\left((S^{j,M}_{s})^{2p}(\sigma(s,I^{M}_{s})-\sigma(s,I_{s}))^{2p}\right)ds.\\[3mm]
  &&\ds +2^{2p-1}
  \int_{0}^{T} \ee\left((\sigma(s,I_{s}))^{2p}(S^{j,M}_{s}-S^{j}_{s})^{2p}\right)ds\\[5mm]
&\leq& \ds 2^{2p-1} C_{2p}^{\frac{1}{2}} K_{Lip}^{2p}T \sqrt{\ee\left(\sup_{0\leq t \leq T}|I^{M}_t-I_t|^{4p}\right)}\\[3mm]
  &&\ds +2^{2p-1} K_b^{2p} \int_0^T \ee\left(\sup_{0\leq t \leq T}|S^{j,M}_{s}-S^{j}_{s}|^{2p}\right) ds\\[5mm]
\end{array}\]
We deduce using Gronwall's lemma :
\[\ee\left(\sup_{0\leq t \leq T}|S^{j,M}_t-S^j_t|^{2p}\right) \leq
\widetilde{C}^j_T \sqrt{\ee\left(\sup_{0\leq t \leq
T}|I^{M}_t-I_t|^{4p}\right)}
\] where
\[\widetilde{C}^j_{T}=6^{2p-1}K_pT^p\beta_{j}^{2p}C_{2p}^{\frac{1}{2}} K_{Lip}^{2p} \,\, e^{3^{2p-1}((r-\delta_{j})^{2p}T^{2p-1}+K_pT^{p-1}K_\eta^{2p}+
2^{2p-1}K_pT^{p-1}\beta_{j}^{2p}K_b^{2p})T}.\] We conclude by
Theorem \ref{convind} and the sublinearity of the square root
function on ${\mathbb R}_+$.

$\,$\\
We now turn to the $L^{2p}$-distance between $I^M$ and
$\overline{I}^M$ :
\[\begin{array}{rcl}
\ds |I^M_t - \overline{I}^M_t|^{2p} & = &\ds \left|\sum_{j=1}^M w_j
S^{j,M}_t-\sum_{j=1}^M w_j S^{j}_t\right|^{2p}\\[3mm]
&\leq& \ds \left(\sum_{j=1}^M w_j |S^{j,M}_t-S^{j}_t|\right)^{2p}\\[3mm]
&\leq& \ds \sum_{j_1=1}^M \dots \sum_{j_{2p}=1}^M \prod_{k=1}^{2p}
w_{j_k} |S^{j_k,M}_t-S^{j_k}_t|\\
\end{array}\]
So, using Hölder inequality, one has
\[\begin{array}{rcl}
\ds \ee\left(\sup_{0\leq t\leq T} |I^M_t -
\overline{I}^M_t|^{2p}|\right) & \leq&\ds \sum_{j_1=1}^M \dots
\sum_{j_{2p}=1}^M \left(\prod_{k=1}^{2p} w_{j_k}\right)
\prod_{k=1}^{2p} \left(\ee(\sup_{0\leq t\leq
T}|S^{j_k,M}_t-S^{j_k}_t|^{2p})\right)^{\frac{1}{2p}}\\[5mm]
&\leq &\ds \left(\sum_{j=1}^{M} w_{j}\right)^{2p} \max_{1\leq j \leq
M}\widetilde{C}^j_T \left(\left(\sum_{j=1}^{M}
w_{j}^2\right)^{\!\!p} + \right.\\[5mm]
&& \ds \quad \left.\left(\sum_{j=1}^M
  w_j|\beta_j-\beta|\right)^{2p}+ \left(\sum_{j=1}^M
w_j|\delta_j-\delta|\right)^{2p} \right).
\end{array}\]

\end{proofOFT2}

\end{document}